\documentclass[11pt]{article}

\usepackage{amsfonts}
\usepackage{amssymb}
\usepackage{amsthm}
\usepackage{epsf}
\usepackage{pstricks,pst-node,pst-text,pst-3d}
\usepackage{amsmath,stmaryrd,amsthm}
\usepackage[mathscr]{eucal}
\usepackage{tikz}
\usepackage[english]{babel}
\usepackage{subfig}
\usepackage{mathtools}
\usepackage{algorithm}
\usepackage[noend]{algpseudocode}
\usepackage{bm}
\usepackage[numbers]{natbib}
\usepackage{booktabs}
\usepackage{ulem}
\usetikzlibrary{decorations.pathreplacing}
\usepackage{sectsty}
\usepackage{xcolor}
\usepackage{hyperref}
\usetikzlibrary{graphs, graphs.standard}

\sectionfont{\fontsize{14}{15}\selectfont}

\def\cA {\mathscr{A}}
\def\cB {\mathscr{B}}
\def\cC{\mathscr{C}}
\def\cF{\mathscr{F}}

\def\cBG {\mathscr{BG}}
\def\cG{\mathscr{G}}
\def\gB{\mathfrak{B}}
\def\gC{\mathfrak{C}}
\def\gD{\mathfrak{D}}

\def\gM{\mathfrak{M}}
\def\RR{{\mathbb R}}
\def \mymid {\,:\,}
\def\Lin{\mathsf{Lin}}
\def \proj {\mathrm{proj}}

\makeatletter
\DeclareRobustCommand{\rvdots}{%
  \vbox{
    \baselineskip4\p@\lineskiplimit\z@
    \kern-\p@
    \hbox{.}\hbox{.}\hbox{.}
  }}
\makeatother

\theoremstyle{definition}
\newtheorem{definition}{Definition}
\newtheorem{theorem}{Theorem}
\newtheorem{proposition}{Proposition}
\newtheorem{remark}{Remark}
\newtheorem{example}{Example}
\newtheorem{corollary}{Corollary}
\newtheorem{lemma}{Lemma}

\begin{document}

\title{\Large\bf On the closest balanced game}

\author{Pedro GARCIA-SEGADOR${}^{a}$,  Michel
    GRABISCH${}^{b,c}$,\\ Dylan LAPLACE MERMOUD${}^{d}$ and Pedro MIRANDA${}^{e}$\\
\normalsize   ${}^a$ National Statistical Institute, Madrid, Spain \\
\normalsize ${}^b$Charles University, Prague, Czech Republic\\
\normalsize  ${}^c$ Universit\'e Paris I Panth\'eon-Sorbonne, Centre d'Economie
de la Sorbonne, Paris, France\\
\normalsize ${}^d$ CNRS, LIP6, Sorbonne Universit\'e, Paris, France\\
\normalsize  ${}^e$ Interdisciplinary Mathematical Institute, Complutense University of Madrid, Spain}

\date{}

\maketitle

\begin{abstract}
Cooperative games with nonempty core are called balanced, and the set of
balanced games is a polyhedron. Given a game with empty core, we look for the closest balanced game, in the sense of the
(weighted) Euclidean distance, i.e., the orthogonal projection of the game on the set of balanced games. Besides an analytical approach which
becomes rapidly intractable, we propose a fast algorithm to find the closest
balanced game, avoiding exponential complexity for the optimization problem,
and being able to run up to 20 players. Our simulations suggest that the probability that the closest game has a
core reduced to a singleton increases towards one as the number of players
grows. We provide a mathematical proof that the proportion of
facets whose games have a non-singleton core tends to 0 when the number of
players grow, by finding an expression of the asymptotic
growth of the number of minimal balanced collections.
We also show that the proportion of faces containing games with a
non-singleton core tends to zero. Under a precise condition limiting
concentration of projection probabilities on individual faces, these
combinatorial results imply that the probability of a singleton core for
the projected game tends to one. Consequently, taking the core of the
projected game defines a new solution concept, which we call least square
core due to its analogy with the least core.
\end{abstract}


\section{Introduction}
Cooperative game theory has been introduced by von Neumann and Morgenstern
\cite{vnm44} to model situations where players, by forming coalitions, may
create larger benefits or reduce costs. Assuming that all players decide to
cooperate and form the grand coalition $N$, the central question is how to share
in a rational way the total benefit, denoted by $v(N)$, among all the players in
$N$. Any such sharing method is called a solution concept, and one of the most
famous is the core. The notion of the core was introduced by Gillies \cite{gil53} in his study of stable sets,
and later popularized by Shapley as a fundamental solution concept (see \cite{shubik_mem,zhao_core}). The key feature of the
core is that any payment vector in the core satisfies coalitional rationality:
any coalition $S$ will receive at least the amount equal to the benefit the
coalition can produce by itself, that is, $v(S)$. Unfortunately, it is not
always possible to satisfy this condition for every coalition $S$, making the
core an empty set. The well known result of Bondareva \cite{bon63} and Shapley
\cite{sha67} gives a necessary and sufficient condition for a game to have a
nonempty core: these games are called {\it balanced}. In a companion paper
\cite{gagrmi25}, the authors have investigated in detail the geometrical
structure of the set of balanced games (see also a recent paper by Abe and
Nakada \cite{abna23}), and
found that it is a nonpointed cone.

To cope with the situation of having an empty core, many other solution concepts
based on different principles have been proposed (e.g., the Shapley value, the
nucleolus, etc.), but also variations of the core, like the $\epsilon$-core and
the least core \cite{shsh66}. The idea of the $\epsilon$-core is
very simple: just reduce $v(S)$ by a quantity $\epsilon$ for all coalitions $S$,
until satisfying coalitional rationality. The least core
is then the $\epsilon$-core with smallest $\epsilon$ keeping the $\epsilon$-core
nonempty (hence, the smallest $\epsilon$ is positive iff the core is empty).

There is an alternative idea to this equal reduction of $v(S)$ for all $S$,
which generalizes the $\epsilon$-core: why
not to allow unequal reduction of $v(S)$, by a quantity $\epsilon(S)$? Of course, this must be done
in a minimal way, like for the least core: we may for example minimize the (weighted) sum of all $\epsilon(S)$, or
the (weighted) sum of their square.

It turns out that minimizing the sum of the $\epsilon(S)^2$ amounts to
orthogonally projecting the game $v$ on the set of balanced games, equivalently, to find
the closest balanced game $v^*$ to $v$ in the sense of the Euclidean distance,
and to take the core of $v^*$ as a solution for $v$. The
present paper precisely deals with this question.

Projecting on a polyhedron is in general not feasible in an analytical
way. Based on the results of the companion paper \cite{gagrmi25}, we nevertheless provide
explicit expressions for $n=3$ players. For a larger number of players, only the
optimization approach (minimizing the Euclidean distance to the set of balanced
games) is feasible. A naive approach leads however to a quadratic program with
an exponential number of variables and a superexponential number of constraints,
even unknown for more than 7 players. By introducing $n$ auxiliary variables,
the number of constraints can be reduced to an exponential number, which allows
to compute the closest game up to 12 players.

The first main achievement of the paper is to prove that the afore mentionned
optimization problem can be reduced to a problem using only $n$ variables and
one constraint, at the price of having an objective function with a logical
operator. We propose a simple algorithm to solve this optimization problem,
allowing to compute the closest game up to 20 players in a few seconds. As a
consequence, the method and the solution concept it induces is operational for
most games in practice.

In the companion paper \cite{gagrmi25}, we have studied under which conditions a balanced game
has a core reduced to a single point. It never happens for games in the interior
of the cone of balanced games, and we know exactly on which faces and facets
games have a core reduced to a singleton. A simulation study suggests that, as $n$ increases, the probability
that the closest game has a core reduced to a singleton approaches one. This reinforces the interest of the solution
concept we propose (core of the closest balanced game), since most of the time a
single point is obtained instead of a set. Our second main achievement is to establish rigorous combinatorial
results and a conditional probabilistic theorem explaining this phenomenon.
We prove that the proportion of facets
whose games have a singleton core tends to 1 when $n$ tends to infinity. This is achieved by giving an expression of the asymptotic growth of the
  number of minimal balanced collections (generalization of partitions) and the number of minimal balanced
  collections of size $n$.  As very few results exist on minimal  balanced
  collections, this result brings an important contribution to Combinatorics.
Then, we
prove that the same result holds also for faces. Under the NEFPC
concentration condition, we finally prove that the probability that the
closest game has a singleton core tends to 1.

The third achievement of the paper is to have proposed a new solution concept,
similar to the least core, which we may call the {\it least quadratic
    distance core} or {\it least square core}. However, the present paper
concentrates on the projection problem and the fact that most of the time the
core of the projected game has a singleton core. A thorough study of this new
solution concept, in the traditional way of cooperative game theory, still
remains to be done, and will be the topic of a future paper.

The paper is organized as follows. Section~\ref{sec:baco} gives the basic
material on balanced games and recalls the results on the set of balanced
games. Section~\ref{sec:ficl} studies the problem of finding the closest
balanced game, both by the analytical way (Section~\ref{sec:anwa}) and the
optimization way (Section~\ref{sec:opwa}). Section~\ref{sec:clob} introduces the
CLOBIS algorithm, used to find in an efficient way the closest balanced
game. Section~\ref{sec:simu} is devoted to the simulations, in order to study
the probability that the closest game has a singleton core
(Section~\ref{sec:prsi}) and the probability of facets and faces containing games
with a singleton core (Section~\ref{sec:fafa}). Lastly, Section~\ref{sec:prob}
proves in a theoretical way the findings of the previous section and Section 6 presents our conclusions.

\section{Basic concepts}\label{sec:baco}

Throughout the paper we consider a (fixed) set $N$ of $n$ players, denoted by
$N= \{ 1, \ldots , n\} $. {\it Coalitions} are nonempty subsets of $N$, and will
be denoted by capital letters $S,T$, etc. A {\it TU-game} $(N, v)$ (or simply a
game $v$) is a function $ v: 2^N \rightarrow \mathbb{R}$ satisfying
$v(\varnothing )=0.$ The value $v(S)$ represents the maximal value (benefit)
that the coalition $S$ can guarantee, no matter what players outside $S$ might
do. We will denote by $\cG(n)$ the set of games $v$ on $N.$ For further use, we
introduce $\cG_\alpha(n)$, $\alpha\in\RR$, the set of games $v\in \cG(n)$
such that $v(N)=\alpha$. In addition, we will often use the following family of
games, which are bases of the $(2^n-1)$-dimensional vector space $\cG(n)$:
The {\it Dirac games} $\delta_S$, $\varnothing\neq S\subseteq N$, defined by
  \[
  \delta_S(T) = \begin{cases}
    1, & \text{ if } T=S\\
    0, & \text{ otherwise}
    \end{cases}.
    \]


Assuming that all players agree to form the grand coalition $N$, we look for a rational way to share the benefit $v(N)$ among all players, i.e., for an {\it
  allocation} or {\it payment vector} $x\in \RR^N$, where coordinate $x_i$ indicates the payoff given to player $i.$ For any coalition $S$, we denote by
$x(S):=\sum_{i\in S} x_i$ the total payoff given to the players in $S$. An allocation is {\it efficient} if $x(N)=v(N)$. A systematic way of assigning a
set of allocations to a game is called a {\it solution concept}.

Of course, we look for an allocation $x$ being reasonable according a criterion, in the sense that $x_i$ is sensible for player $i.$ There are many solution concepts. In this paper we focus on one of the best known solution concepts, which is the {\it core} \cite{gil53}. The core is the set of efficient allocations satisfying {\it coalitional rationality}, which means that $x(S)\geqslant v(S)$ for all coalitions $S$. Under this condition, no coalition $S$ has an incentive to leave the grand coalition $N$ to form a subgame on $S$. The core of a game $(N,v)$ will be denoted by $C(N,v)$ (or $C(v)$ for short if $N$ is fixed):

\[ C(v):=\{ x\in \mathbb{R}^n: x(S)\geqslant v(S), \forall S\in2^N,x(N)=v(N)\}.\]
The core of a game is a closed convex polytope that may be empty.
%
%
A condition for nonemptiness of the core has been given by Bondareva \cite{bon63} and Shapley \cite{sha67}, which we detail below.

A {\it balanced collection} $\cB$ on $N$ is a family of nonempty subsets of $N$ such that there exist positive (balancing) weights $\lambda^{\cB}_S$, $S\in\cB$, satisfying
\begin{equation}\label{eq:bc}
  \sum_{\substack{S\in\cB\\S\ni i}} \lambda^{\cB}_S = 1, \forall i\in N.
\end{equation}
This notion is an extension of the notion of partition, as any partition is a
balanced collection with balancing weights all equal to 1. For a balanced
collection  $\cB =\{ S_1, ..., S_r\} ,$ we can build a $ n\times r$ matrix $M_{\cB}$ via
$$ M_{\cB}(i,j)= \left\{ \begin{array}{cc} 1 & \text{~if~}  i\in S_j \\ 0 & \text{otherwise} \end{array} \right. . $$
A balanced collection is {\it minimal} if it contains no balanced proper subcollection. We denote by $\gB^*(n)$ the set of all minimal balanced collections on a set $N$ of cardinality $n$, excluding the collection $\{ N\}$. It can be shown that minimal balanced collections (abbreviated hereafter by m.b.c.) have a unique set of balancing weights and that their cardinality is at most $n$. In terms of $M_{\cB}$ this is equivalent to say that the system

\begin{equation}\label{eq:mbc}
M_{\cB}{\bm\lambda}=\bm 1
\end{equation}
has a unique and positive solution.

A game $v$ is {\it balanced} if
\begin{equation}\label{bal}
 \sum_{S\in \cB} \lambda_S^{\cB} v(S) \leqslant v(N), \quad \forall \cB \in \gB^*(n).
\end{equation}
The following result holds \cite{bon63,sha67}.

\begin{theorem}\label{th:bal}
Consider a game $(N, v)$. Then, $C(v)\ne \varnothing $ if and only if $v$ is balanced.
\end{theorem}
We will denote by $\cBG(n)$ the set of games in $\cG(n)$ being balanced, i.e., the set of games having a nonempty core.
Hence, applying Theorem~\ref{th:bal} and (\ref{bal}), we get

\begin{equation}\label{eq:bg}
\cBG(n)=\Big\{ v\in \cG(n) \mymid \sum_{S \in \cB} \lambda_S^{\cB} v(S) -v(N) \leqslant 0, \quad \forall  \cB \in \gB^*(n) \Big\} .
\end{equation}
 Another set of interest, which we will consider throughout the paper, is
  the set of balanced games $v\in\cG_\alpha(n)$ for some $\alpha\in\RR$. We
  denote by $\cBG_{\alpha }(n)$ the set of such games, i.e.,
\begin{equation}\label{eq:bgalpha}
\cBG_\alpha (n)=\Big\{ v\in \cG_\alpha(n) \mymid \sum_{S \in \cB} \lambda_S^{\cB} v(S) \leqslant \alpha, \quad \forall  \cB \in \gB^*(n) \Big\} .
\end{equation}
The study of this set is motivated as follows: when searching for the closest
balanced game $v^*$ to a given non-balanced game $v$, it is natural to impose
that $v^*(N)=v(N)$, as this is the total benefit which has to be distributed
among players.  Therefore, the set of interest is $\cBG_{v(N)}(n)$.


Note that $\cBG(n)$ and $\cBG_{\alpha }(n)$ are convex polyhedra. The
structure of these polyhedra has been studied in \cite{gagrmi25}.
  Basically, $\cBG_\alpha(n)$ is an affine cone (i.e., a cone plus a point using Minkowski addition)
  which is not pointed. Its precise structure is given in Section~\ref{sec:anwa}. We
  recall below important results related to its faces.

For the combinatorial analysis in Section~\ref{sec:prob}, it is useful to
recall the geometric structure of the full cone $\cBG(n)$ in addition to
the affine section $\cBG_\alpha(n)$ appearing in the projection problem.
Let $u_{\{i\}}=\sum_{S\ni i}\delta_S$ denote the additive game corresponding
to player $i$. Since adding an additive game preserves balancedness and
translates the core, the map
\[
T_\alpha:\cBG_0(n)\longrightarrow\cBG_\alpha(n),\qquad
T_\alpha(w)=w+\alpha u_{\{n\}},
\]
is an affine isomorphism, and $C(T_\alpha(w))=C(w)+\alpha e_n$.
Moreover,
\[
\cBG(n)\simeq\cBG_0(n)\times\mathbb R,\qquad
v\longmapsto \bigl(v-v(N)u_{\{n\}},v(N)\bigr),
\]
with inverse $(w,t)\mapsto w+t u_{\{n\}}$.
Consequently, faces (including facets) of the full cone are in
bijection with the faces of every affine section, and the singleton-core
property is preserved by this correspondence. We may thus count faces
using $\cBG(n)$ even though the closest-game problem is posed on
$\cBG_{v(N)}(n)$.

For a convex cone
$K$, its lineality space is the largest linear subspace contained in $K$,
that is,
\[
\Lin(K)=K\cap(-K).
\]

The geometric structure of the cone $\cBG(n)$ was fully characterized in \cite{gagrmi25}. Specifically, $\cBG(n)$ is a non-pointed polyhedral cone whose lineality space and extremal rays are given by the following results:
\begin{theorem}\label{th:3}
Let $n\geqslant 2$. Then, $\cBG(n)$ is a $(2^n-1)$-dimensional polyhedral cone,
which is not pointed. Its lineality space has dimension $n$, with basis
$(w_i)_{i\in N}$, with $w_i=\sum_{S\ni i}\delta_S$, where $\delta_S$ is the
Dirac game centered on $S$.
\end{theorem}
\begin{theorem}\label{th:4}
Let $n\geqslant 2$. The extremal rays of $\cBG(n)$ are
  \begin{itemize}
  \item The $2n$ extremal rays corresponding to $\Lin(\cBG(n))$:
    $w_1,\ldots,w_n,-w_1,\ldots,-w_n$;
  \item $2^n-n-2$ extremal rays of the form $r_S=-\delta_S$, $S\subset N$, $|S|>1$;
  \item $n$ extremal rays of the form
    \[
r_i = \sum_{S\ni i, |S|>1}\delta_S, \quad i\in N.
    \]
  \end{itemize}
This yields in total $2^n+2n-2$ extremal rays.
\end{theorem}

\begin{theorem}\label{th:bgfacet}
  Each inequality in (\ref{eq:bg}) defines a facet of $\cBG(n)$, and
  each inequality in (\ref{eq:bgalpha}) defines a facet of
  $\cBG_\alpha(n)$. Thus, minimal balanced collections in
  $\gB^*(n)$ are in bijection with the facets of either polyhedron.
\end{theorem}

\begin{theorem}\label{th:pcbg}
Consider a minimal balanced collection $\cB\in\gB^*(n)$ and its corresponding
facet $\cF$ in $\cBG(n)$ (or its corresponding facet in
$\cBG_\alpha(n)$). The following holds:
\begin{enumerate}
\item If $|\cB|=n$, every game in $\cF$ has a point core.
\item Otherwise, no game in the relative interior of $\cF$ has a point core.
\end{enumerate}
\end{theorem}

\begin{theorem}\label{th:pcbg1}
Consider a face $\cF=\cF_1\cap\cdots\cap\cF_p$ of either $\cBG(n)$
or $\cBG_\alpha(n)$, where $\cF_1,\ldots,\cF_p$ are the facets associated
with the minimal balanced collections $\cB_1,\ldots,\cB_p$, respectively.
Then every game in $\cF$ has a point core if and only if the rank of
$\{\mathbf{1}^S:S\in\cB_1\cup\cdots\cup\cB_p\}$ is $n$, where
$\mathbf{1}^S\in\mathbb R^n$ is the incidence vector of $S$.
\end{theorem}

\section{Finding the closest balanced game}\label{sec:ficl}
Let us look at a numerical example:

\begin{example}\label{Example1}
A group of four companies in the same sector has analyzed the possibility of collaborating and decided that cooperation among all four companies would be beneficial for everyone. They conducted a market study and obtained an estimate of the maximum profits each coalition would achieve to ensure a fair distribution.
The estimated game is given in Table \ref{Tab:1}.  For simplicity, subsets are
  written without braces and commas, i.e., $\{1,2,3\}$ is denoted by 123. This
  convention will be used throughout the paper.
\begin{table}[h]
\centering
\begin{tabular}{|c|cccccccccc|}
\hline $S$ & $1$ & $2$ & $3$ & $4$ & $12$ & $13$ & $14$ & $23$ & $24$ & $34$ \\
$v(S)$     & 32      & 25      & 27      & 16      & 46       & 46       & 38       & 49       & 26       & 54   \\
\hline $S$ &  $123$ & $124$ & $134$ & $234$ & $1234$ & & & & & \\
    $v(S)$ &         95 &        79 &        64 &        88 & 100 & & & & & \\
\hline
\end{tabular}
\caption{Estimated game $v$ for Example \ref{Example1}.}
\label{Tab:1}
\end{table}

A natural benefit that each of the companies expects to obtain could be the one provided by an imputation in the core of the game. However, it can be checked that the game has an empty core. So, what can be done?
\end{example}


Consider a game $v$ with an empty core. In order to find a ``reasonable'' way to
share the value of the grand coalition $v(N),$ we propose to look for the
closest game $v^*$, according to the (possibly weighted) Euclidean
distance to the set $\cBG_\alpha(n)$ with $\alpha=v(N)$, and take
one of the imputations in $C(v^*).$  As it will become apparent, this
  amounts to reducing in a minimal way some quantities $v(S)$, $S\subset N$.

   Finding the closest game $v^*$ w.r.t. the Euclidean distance amounts to
    projecting $v$ on  the set $\cBG_{v(N)}(n)$. This
    shows that the problem can be solved either in an analytical way, using
    results of linear algebra, or by solving an optimization problem. We address
    these two ways in what follows.

 \subsection{The analytical way}\label{sec:anwa}
While orthogonally projecting a point on an affine or linear subspace is a routine exercise,
there is no analytical expression of an orthogonal projection on a polyhedron. The reason is
that we do not know a priori on which facet or face the projection will
fall. This is well illustrated by Fig.~\ref{fig:mosaico}, where the polyhedron
is simply a triangle.
\begin{figure}[h]
\centering
\begin{tikzpicture}[scale=4]

\pgfmathsetmacro{\h}{sqrt(3)/2}

\coordinate (A) at (0,0);
\coordinate (B) at (1,0);
\coordinate (C) at (0.5,\h);

\clip (-0.7,-0.5) rectangle (1.7,1.2);

\fill[yellow!40, opacity=0.4] (A)--(B)--(C)--cycle;

\fill[blue!40, opacity=0.4] (0.0,-0.5) -- (1.0,-0.5) -- (1.0,0) -- (0.0,0) -- cycle;

\fill[green!40, opacity=0.4] (0,0) -- (0.5,\h) -- (-0.08,1.2) -- (-0.7,1.2) -- (-0.7,0.4) -- cycle;

\fill[pink!40, opacity=0.4] (1,0) -- (0.5,\h) -- (1.08,1.2) -- (1.7,1.2) -- (1.7,0.4) -- cycle;

\fill[green!80, opacity=0.4] (0,0) -- (0,-0.5) -- (-0.7,-0.5) -- (-0.7,0.4) -- cycle;

\fill[purple!60, opacity=0.4] (1,0) -- (1,-0.5) -- (1.7,-0.5) -- (1.7,0.4) -- cycle;

\fill[orange!60, opacity=0.4] (0.5,\h) -- (-0.08,1.2) -- (1.08,1.2) -- cycle;

\draw[thick] (A)--(B)--(C)--cycle;

\foreach \p in {A,B,C}{
    \fill[red] (\p) circle[radius=0.015];
}

\end{tikzpicture}
\caption{Projecting on a triangle. Each region shows the set of points projecting on each face (edges and vertices) of the triangle.}\label{fig:mosaico}
\end{figure}
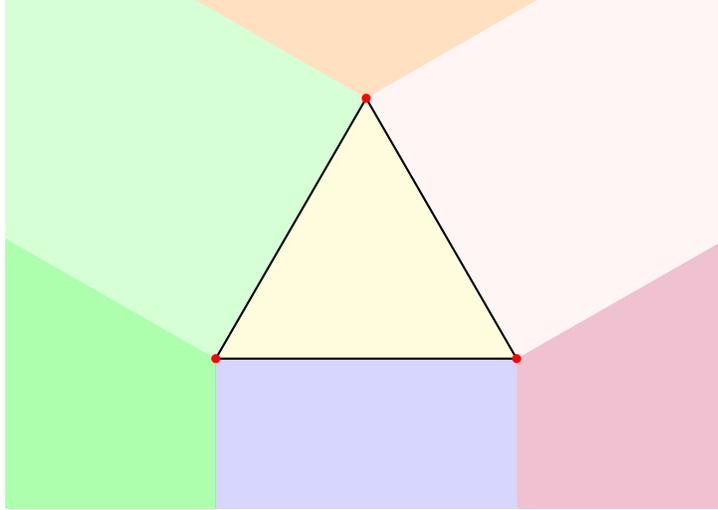

We clearly see that depending on which zone the point lies, its orthogonal projection will
fall either on one of the edges or on one of the vertices. However, the inverse
problem is easier: for a given face of $\cBG(n)$ or $\cBG_\alpha(n)$, find all games whose orthogonal
projection falls on that face. We address this problem for the case of
$\cBG_\alpha(n)$.

We recall from \cite{gagrmi25} that $\cBG_\alpha(n)$ is an
affine cone, expressed by
\[
\cBG_\alpha(n)=\alpha u_{\{n\}}+\cC_0(n),
\]
with $u_{\{n\}}=\sum_{S\ni n}\delta_S$, and $\cC_0(n)$ is a cone whose expression
  is
  \[
\cC_0(n) = \Lin(\cC_0(n))+\cC^0_0(n),
\]
where $\cC_0^0(n)$ is a pointed cone such that games in it have coordinates
0 for the subsets $\{1\},\ldots\{n-1\}$, and $\Lin(\cC_0(n))$ is its lineality
space, with basis $w_1,\ldots,w_{n-1}$ given by
\[
w_i = \sum_{\substack{S\ni i\\S\not\ni n}}\delta_S  -  \sum_{\substack{S\not\ni
    i\\S\ni n}}\delta_S, \ i=1,\ldots,n-1.
\]
Moreover, the extremal rays of $\cC_0(n)$ are, apart the vectors of the basis of
the lineality space:
  \begin{itemize}
  \item  $2^n-n-2$ extremal rays of the form $r_S=-\delta_S$, $S\subset N$, $|S|>1$;
  \item $n$ extremal rays of the form
    \[
r_i = \sum_{\substack{S\ni i\\ S\not\ni n\\|S|>1}}\delta_S-\sum_{\substack{S\not\ni i\\ S\ni n}}\delta_S,\quad i\in N\setminus\{n\},
\]
and $r_n=-\delta_{\{n\}}$.
  \end{itemize}

Note that $v(N)=0$ for any game on $\cC_0(n).$ From standard results, the projection of some game $v$ on $\cBG_\alpha(n)$ writes
\[
v^*=\proj_{\cBG_\alpha(n)}(v) = \alpha u_{\{n\}} + \proj_{\cC_0(n)}(v-\alpha u_{\{n\}}).
\]

Let us consider a facet of $\cBG_\alpha(n)$ (translation by $\alpha u_{\{n\}}$
of a facet on $\cC_0(n)$), which corresponds to a minimal balanced
collection, say $\cB$ (with some abuse, we denote the facet by the same symbol).
Let us denote by $\rho_1,\ldots,\rho_r$ the set of extremal rays in $\cB$ not in
the lineality space. As the lineality space is included in any facet, any $v^*$
in the facet $\cB$ writes
\[
v^* = \alpha u_{\{n\}}+\sum_{i=1}^r\alpha_i\rho_i + \sum_{i=1}^{n-1}\beta_iw_i
\]
with $\alpha_i\geqslant 0$, $i=1,\ldots,r$, and $\beta_i\in\RR$, $i=1,\ldots,n-1$.

Any game $v$ such that $v-v^*$ is colinear with the normal vector of the facet
will have as projection on the facet the game $v^*$. As the normal vector is
$\lambda^\cB=(\lambda^\cB_S)_{S\in 2^N\setminus \{\varnothing\}}$ (the vector of
balancing weights, putting 0 when $S\not\in\cB$), we finally obtain that any $v$
whose projection falls on $\cB$
must have the form:
\[
v=\alpha u_{\{n\}}+\gamma\lambda^\cB+\sum_{i=1}^r\alpha_i\rho_i + \sum_{i=1}^{n-1}\beta_iw_i
\]
for some $\gamma,\alpha_1,\ldots,\alpha_r\geqslant 0$ and some real numbers
$\beta_1,\ldots, \beta_{n-1}$.

The case of faces is similar. A face is an intersection of facets, and is
therefore identified as a collection of m.b.c. $\cB_1,\ldots,\cB_p$. Any game
$v^*$ in this face is a conic combination of all extremal rays
$\rho_1,\ldots,\rho_r$ belonging to all facets $\cB_1,\ldots,\cB_p$, plus a game
in the lineality space. Therefore, any game whose projection falls on that face
has the form:
\[
v=\alpha u_{\{n\}}+\sum_{i=1}^p\gamma_i\lambda^{\cB_i}+\sum_{i=1}^r\alpha_i\rho_i + \sum_{i=1}^{n-1}\beta_iw_i
\]
for some $\gamma_1,\ldots,\gamma_p,\alpha_1,\ldots,\alpha_r\geqslant 0$ and some real numbers
$\beta_1,\ldots, \beta_{n-1}$.

Lastly, all games which are not of the above form for any face are projected on
the lineality space $\Lin(\cBG(n))$. Let us find the game $v^*$ orthogonal projection of $v$ in the previous conditions in the lineality space.
Remark that $w_i=u_{\{i\}} - u_{\{n\}}$, $i=1,\ldots,n-1$. It is
convenient to use the M\"obius representation $m^v$ for a game $v$\footnote{Recall that the M\"obius representation (a.k.a. Harsanyi dividends) is the
coordinates of a game in the basis of unanimity games $(u_S)_{S\subseteq
  N,S\neq\varnothing}$, with $u_S$ defined by $u_S(T)=1$ if $T\supseteq S$ and 0
otherwise. Formally, $v=\sum_{S\subseteq N,S\neq \varnothing}m^v(S)u_S$,
which explains the property $v(N)=\sum_{S\subseteq N,S\neq\varnothing}m^v(S)$.} In
$\cC_0(n)$, we have $v(N)=0$, therefore $\sum_{S\subseteq N}m^v(S)=0$ and we can
express one coordinate in terms of the others, for example
$m^v(\{n\})=-\sum_{S\subseteq N,S\neq\{n\}}m^v(S)$. Doing so, the basis
$w_1,\ldots,w_{n-1}$ becomes orthonormal in this system of coordinates (as $m^{w_i}(i)=1, m^{w_i}(S)=0, S\ne i$). We can
then conveniently express the orthogonal projection $v^*_0$ of a game $v$ on
$\Lin(\cC_0(n))$:
\[
m^{v^*_0} = \sum_{i=1}^{n-1}\langle m^v,m^{w_i}\rangle m^{w_i} = \sum_{i=1}^{n-1}m^v(\{i\})m^{w_i},
\]
which yields $m^{v_0^*}(\{i\})=m^v(\{i\})$ for $i=1,\ldots, n-1$, $m^{v_0^*}(S)=0$
for $|S|>1$, and $m^{v^*_0}(\{n\})=-\sum_{i=1}^{n-1}m^v(\{i\})$. Finally, the
projection on $\Lin(\cC_0(n))$ is an additive game determined by
\[
v^*_0(\{i\}) = \begin{cases}
  v(\{i\}), & \text{ if } i=1,\ldots,n-1\\
  -\sum_{j=1}^{n-1}v(\{j\}), & \text{ if } i=n.
  \end{cases}
\]
Now, the projection of $v$ on $\cBG_{\alpha }(n)$ is the additive game
$v^*=v^*_0+\alpha u_{\{n\}}$:
\[
v^*(\{i\}) = \begin{cases}
  v(\{i\}), & \text{ if } i=1,\ldots,n-1\\
  \alpha-\sum_{j=1}^{n-1}v(\{j\}), & \text{ if } i=n.
  \end{cases}
\]

The whole procedure can be applied for $n=3$, since the number of faces  is
equal to 20, and is detailed in
Appendix~\ref{app:n3}. For $n>3$, the procedure is no more feasible as the
number of faces grows too fast (see e.g. Table \ref{tablembmc} for the growth of facets when $n$ increases).

\subsection{The optimization way}\label{sec:opwa}
 The weights
$\gamma_S >0, S\subset N$, measure the importance of each coalition in terms of $v,$
in the sense that the greater $\gamma_S,$ the more interested we are in
obtaining a game $v^*$ satisfying $v^*(S)$ is close to $v(S).$  Note that
  the case with equal weights amounts to taking the orthogonal projection.
Formally, we aim to solve
\begin{equation}
\begin{array}{ll}
 \min & {\displaystyle \sum_{S \subseteq N} \gamma_S(v(S) - w(S))^2} \\
\text{s.t.} &  w \in \cBG_{v(N)}(n).
\end{array}
\label{eq:optimization_problem1}
\end{equation}
As $\cBG_{v(N)}(n)$ is a convex polyhedron, there exists a unique solution
$v^*$ for this problem, and the constraints are linear. Therefore,
(\ref{eq:optimization_problem1}) is a quadratic programming
problem. There are several procedures to solve such type of problems, however we
face a  large number of
constraints. Indeed, there is a constraint for each minimal balanced collection
$\cB \in \gB^*(n)$, and the number of minimal balanced collections
grows very quickly with $|N|$. Table \ref{tablembmc} shows the number of minimal
balanced collections for the first few values of $n$ \cite{lagrsu23}.\footnote{This sequence corresponds to sequence A355042 in the Online Encyclopedia of Integer Sequences (OEIS).}
\begin{table}[h]
\begin{center}
\begin{tabular}{|c|r|}
\hline $n$  & Number of minimal balanced collections \\
\hline 1    & 1 \\
       2    & 2 \\
       3    & 6 \\
       4    & 42 \\
       5    & 1,292 \\
       6    & 200,214 \\
       7    & 132,422,036 \\
\hline
\end{tabular}
\end{center}
\caption{The number of minimal balanced collections for the first values of $n$.}
\label{tablembmc}
\end{table}

Furthermore, these constraints depend on the minimal balanced collections, and obtaining all these collections is a very complex problem. An algorithm for generating minimal balanced collections was proposed in \cite{pel65} and recently extended in \cite{lagrsu23}, but to our knowledge no general characterization of m.b.c. suitable for our purposes has been derived.

Hence, in practice it is necessary to look for another way to obtain $v^*$,
reducing the number of constraints. Notice that these constraints are meant to
impose that $v^*$ is a balanced game. As a balanced game has a nonempty core,
there must exist some $x\in C(v^*),$ i.e., some $x\in\RR^N$ such that $x(S)\geqslant
v^*(S), x(N)=v^*(N).$ Consequently, Problem (\ref{eq:optimization_problem1}) can
be equivalently rewritten as
\begin{equation}
\begin{array}{ll}
\min & {\displaystyle f(w,x)=\sum_{S \subseteq N} \gamma_S (v(S) - w(S))^2,} \\
\text{s.t.} &  w(S) \leqslant x(S), \ \forall S \subset N, \\
&  w(N) = x(N) = v(N).
\end{array}
\label{eq:optimization_problem2}
\end{equation}
where the variables are $w(S)$, $S \subseteq N$, and $x\in\RR^n$. Hence, we have $2^n - 2 + n - 1=2^n+n-3$ variables and $2^n - 1$ constraints, far less than the number of minimal balanced collections.

Note that although there is unique optimal $v^*$, the allocation $x$ might not
be so,
as it does not appear in the objective function of (\ref{eq:optimization_problem2}). Therefore, any allocation in
$C(v^*)$ is valid to achieve the optimum and we already know  from
Theorems~\ref{th:pcbg} and \ref{th:pcbg1} that the core of
games in the border of $ \cBG_{\alpha
}(n)$ might not be a set with a single element. Consequently, the solution of
this problem might not be unique.

\begin{example}
We continue with Example \ref{Example1}. We solve Problem
(\ref{eq:optimization_problem2}) with $\gamma_{S} = 1, \forall S\subset N$ and obtain the game $v^*$ given in Table \ref{tab:2}.

\begin{table}[h]
\centering
\begin{tabular}{|c|cccccccccc|}
\hline $S$ & $1$ & $2$ & $3$ & $4$ & $12$ & $13$ & $14$ & $23$ & $24$ & $34$ \\ \hline
$v(S)$     & 32      & 25      & 27      & 16      & 46       & 46       & 38       & 49       & 26       & 54   \\
$v^*(S)$   & 22.26   & 25      & 27      & 15.21   & 46       & 46       & 37.47     & 49       & 26       & 45.10 \\
\hline $S$ &  $123$ & $124$ & $134$ & $234$ & $1234$ & & & & & \\ \hline
    $v(S)$ &         95 &        79 &        64 &        88 & 100 & & & & & \\
  $v^*(S)$ & 84.79      & 70.10     & 64     & 77.74     & 100  & & & & & \\
\hline
\end{tabular}
\caption{Game $v$ and its closest balanced game $v^*$.}
\label{tab:2}
\end{table}
  We can observe that the values of the closest balanced game $v^*$ are lower than those of the original
game. This is a general fact that will be proved in Lemma~\ref{lem:low}. The objective function ensures that
the values of $v(S)$ are diminished as few as possible to get a balanced game.

For $v^*$, there is just one imputation in $C(v^*)$ given by

$$x_1=22.26, x_2=32.64, x_3=29.89, x_4=15.21.$$
\end{example}

Let us now treat briefly about the computational complexity of Problem
(\ref{eq:optimization_problem2}). Standard algorithms for convex quadratic
programming, such as interior-point methods, have computational complexity
which is polynomial in the number of variables and constraints, denoted by $p.$ It
has been shown that there are algorithms with complexity $O(k p^3),$ where $k$
is a constant. However, since in this case $p=2^n+n-3 + 2^n-2,$ the complexity for
Problem (\ref{eq:optimization_problem2}) becomes exponential $O(8^n)$ on the
number of players. As the number of variables and constraints still grow
exponentially, we obtain a problem computationally intensive and potentially
infeasible to solve exactly for large values of $n.$ Note on the other hand that
having exponential complexity is common in Cooperative Game Theory, as the very nature of
the problem leads the number of coalitions growing exponentially. In Table
\ref{tab:execution_times} we show the average execution time (taken over 100
realizations) for solving this
problem for different values of $n$.

\begin{table}[h]
\centering
\begin{tabular}{|c|r|}
\hline
$n$ & Average Time (s) \\
\hline
3  &  0.0007291489 \\
4  &  0.0007600629 \\
5  &  0.0009162869 \\
6  &  0.0049808989 \\
7  &  0.0083431290 \\
8  &  0.0502450010 \\
9  &  0.5814768980 \\
10 & 2.5622723129 \\
11 & 51.0537700000 \\
12 & 178.8491900000 \\
\hline
\end{tabular}
\caption{Average execution time for computing the closest balanced game as a function of $n$.}
\label{tab:execution_times}
\end{table}

As we can see in Table \ref{tab:execution_times}, the time calculation of the closest balanced game using the previous algorithm grows exponentially in $n.$  In next subsection we will develop another way to deal with this problem.

\subsection{CLOBIS algorithm}\label{sec:clob}

First, let us explore some properties of $v^*$ and its core, considering the
general case $v\in\cG(n)$. Recall that $v^*(N)=v(N).$

\begin{lemma}\label{lem:low}
Let $v$ be a game and $v^*$ be its closest balanced game according to weights $\gamma_S, S\subset N.$ Then, $v^*(S)\leqslant v(S), \ \ \forall S\subset N.$
\end{lemma}

\begin{proof}
Suppose there exists $A\subset N$ such that $v^*(A)> v(A).$ Let $x \in C(v^*).$ Hence, $x(A)\geqslant v^*(A).$ Consider the game $v'$ given by
$$ v'(B)=\begin{cases} v(A) & \text{ if } B=A \\ v^*(B) & \text{ otherwise.} \end{cases} $$
Then, $v'$ is a balanced game because $x \in C(v')$ and

$$ \sum_{B\subset N}\gamma_B (v'(B)- v(B))^2 < \sum_{B\subset N}\gamma_B (v^*(B)- v(B))^2,$$ a contradiction.
\end{proof}

Now, for $v^*$, we have the following properties on any allocation $x\in C(v^*).$

\begin{lemma}\label{closest2}
Let $v$ be a game and $v^*$ be its closest balanced game according to weights $\gamma_S, S\subset N.$ Consider $x\in C(v^*)$ and $S\subset N$ such that $x(S)\geqslant v(S).$ Then, $v^*(S)=v(S).$
\end{lemma}

\begin{proof}
Suppose there exists $A\subset N$ such that $x(A)\geqslant v(A)$ and $v^*(A)\ne v(A).$ Then, consider

$$ v'(B):=  \begin{cases} v(A) & \text{ if } B=A \\ v^*(B) & \text{ otherwise.} \end{cases} $$
Thus defined, as $x(A) \geqslant v(A) =v'(A),$ and $x(B) \geqslant v^*(B) = v'(B),$ it follows that $x\in C(v'),$ so that $v'\in \cBG_{v(N)}(n)$. On the other hand,

$$ \sum_{B\subset N} \gamma_B (v(B) - v'(B) )^2 < \sum_{B\subset N} \gamma_B
(v(B) - v^*(B) )^2,$$ contradicting that $v^*$ is the closest balanced game.
\end{proof}

\begin{lemma}\label{closest3}
Let $v$ be a game and $v^*$ be its closest balanced game according to weights $\gamma_S, S\subset N.$ Consider $x\in C(v^*)$ and $S\subset N$ such that $x(S)\leqslant v(S).$ Then, $v^*(S)=x(S).$
\end{lemma}

\begin{proof}
Suppose there exists $A\subset N$ such that $v^*(A)<x(A)\leqslant v(A).$ Then, let us define

$$ v'(B):= \begin{cases} x(A) & \text{ if } B=A \\ v^*(B) & \text{ otherwise.} \end{cases}  $$
Thus defined, it follows that $x\in C(v')$. But

$$ \sum_{B\subset N} \gamma_B (v(B) - v'(B) )^2 < \sum_{B\subset N} \gamma_B (v(B) - v^*(B) )^2,$$ contradicting that $v^*$ is the closest balanced game.
\end{proof}

As a conclusion of these results, the following holds.

\begin{corollary}\label{corclosest}
Consider $v\in \cG(n)$ and let us denote by $v^*$ its closest balanced game with respect to some weights $\gamma_S, S\subset N$, satisfying $v(N)=v^*(N).$ Let $x\in C(v^*).$ Then, $v, v^*$ and $x$ are related by

\begin{equation}\label{closest-game}
v^*(S)=\min \{ x(S), v(S)\} \, \, \forall S\subset N, \quad  v^*(N)=v(N).
\end{equation}
\end{corollary}

 This result highlights the differences between the core of $v^*$ and the least core. Rather than reducing each coalition value $v(S)$ by a fixed $\epsilon$ to ensure a non-empty core, here we adjust each $v(S)$ by different amounts to obtain a nonempty core.

Corollary \ref{corclosest} will be the base of the algorithm that we will develop below as it establishes that $v^*$ can be derived from any imputation $x\in C(v^*).$ The idea is to look for an allocation $x \in C(v^*)$ instead of looking for $v^*$ itself, and apply Eq. (\ref{closest-game}).

\begin{proposition}\label{prop:equiv}
  Consider the following problem in variable $x\in\RR^n$:
\begin{equation}\label{eq:alg_general}
 \begin{array}{cc} \min & {\displaystyle g(x)=\sum_{S \subset N} \gamma_S (v(S)-x(S))^2 {\bm 1}_{x(S)\leqslant v(S)}} \\ s. t. & x(N)=v(N) \end{array}
\end{equation}
where ${\bm 1}$ denotes the indicator function, taking value 1 if $x(S) \leqslant v(S)$ and
0 otherwise. Then, problems (\ref{eq:optimization_problem2}) and
(\ref{eq:alg_general}) are equivalent in the following sense: If $x^*$ is an
optimal solution of (\ref{eq:alg_general}), then $(v^*,x^*)$ is an optimal
solution of (\ref{eq:optimization_problem2}) with $v^*$ given by
(\ref{closest-game}); If $(v^*,x^*)$ is an optimal solution of
(\ref{eq:optimization_problem2}), then $x^*$ is an optimal solution of (\ref{eq:alg_general}).
\end{proposition}
\begin{proof}
Consider $x^*$ an optimal solution of (\ref{eq:alg_general}). By
(\ref{closest-game}), we associate $v^*$ to $x^*$ and  $(v^*,x^*)$ satisfies the constraints in
(\ref{eq:optimization_problem2}). By (\ref{closest-game}) we have:
\[
f(v^*,x^*)=\sum_S\gamma_S(v(S)-v^*(S))^2=\sum_{S:x^*(S)\leqslant v(S)}\gamma_S(v(S)-v^*(S))^2=g(x^*).
\]
Suppose $(v^*,x^*)$ is not optimal for (\ref{eq:optimization_problem2}), but
$(v',x')$ is. Then by Corollary~\ref{corclosest}, $v'$ satisfies
(\ref{closest-game}), so that we get
\[
g(x')=\sum_{S:x'(S)\leqslant v(S)}\gamma_S(v(S)-v'(S)))^2=\sum_S\gamma_S(v(S)-v'(S))^2=f(v',x')<f(v^*,x^*)=g(x^*),
\]
a contradiction.

Conversely, let $(v^*,x^*)$ be an optimal solution of
(\ref{eq:optimization_problem2}). Then $x^*$ is a solution of
(\ref{eq:alg_general}). Suppose $x^*$ is not optimal and let $x'$
be an optimal solution of (\ref{eq:alg_general}). Let us associate
$v'$ to $x'$ via (\ref{closest-game}). Then $(v',x')$ is a solution of
(\ref{eq:optimization_problem2}) and
\[
f(v',x')=\sum_S\gamma_S(v(S)-v'(S))^2=\sum_{S:x'(S)\leqslant v(S)}\gamma_S(v(S)-v'(S))^2=g(x')<g(x^*).
\]
Now, by Corollary~\ref{corclosest}, we have
\[
f(v^*,x^*)=\sum_S\gamma_S(v(S)-v^*(S))^2=\sum_{S:x^*(S)\leqslant v(S)}(v(S)-x^*(S))^2=g(x^*),
\]
a contradiction.
\end{proof}

In this way, we have managed to go from a problem with $2^n+n-3$ variables and $2^n-1$ constraints (Problem (\ref{eq:optimization_problem2})) to a problem with $n$ variables and just one constraint (Problem (\ref{eq:alg_general})).

Note however that in this problem, the sum in the objective function depends on
the indicator function and this indicator depends on $x$ (unknown). Hence, we
have to compare all subfamilies of subsets, thus leading to solve $2^{2^n-2}$
different problems and compare the corresponding solutions, that is not
affordable for large $n.$ Below, we propose an iterative algorithm for dealing
with Problem (\ref{eq:alg_general}) without necessarily solving $2^{2^n-2}$
different problems. We have called this algorithm the CLOsest Balanced game
Iterative Systems algorithm (CLOBIS algorithm)\footnote{The CLOBIS algorithm has been implemented in an R package called CloBalGame, available on GitHub \cite{CloBalGame}.}.  We explain it in the
  subsequent paragraphs.

\begin{algorithm}[h]
\begin{algorithmic}
\caption{CLOBIS algorithm for finding $v^*$}\label{alg: iterative}
\State \textbf{Input:} The game $v,$ and the weights $\gamma_S.$
\State \textbf{Output:} A core imputation $x^*$ of the closest balanced game
$v^*.$
\State  \textbf{Set} $\cA_0=2^N\setminus\{N\}$.

\State \textbf{Step $k\geqslant 0$:} Search for an optimal solution $x^{k+1}$ of

$$ \begin{array}{cc} \min & \phi_k(x)={\displaystyle \sum_{S \in \cA_k}} \gamma_S (x(S)- v(S))^2, \\ s. t. & x(N)=v(N). \end{array} $$

Define

$$\cA_{k+1} = \{ S \in 2^N\setminus\{N\}: x^{k+1}(S)\leqslant v(S)\} .$$

Continue until the {\bf stopping criterion} is met.

\State \textbf{Stopping criterion:} $\cA_{k}= \cA_{k+1}.$
\end{algorithmic}
\label{algo}
 \end{algorithm}

 Solving the problem that appears in each iteration $k$ of the previous algorithm can be done based on the Lagrangian of the problem. The Lagrangian function at $k$-th iteration is given by:

\[ L(x,\lambda) = \sum_{S \in \cA_k} \gamma_S \left( x(S) - v(S)
\right)^2 - \lambda \left( x(N) - v(N) \right).\]
 with $\cA_k=\{S:{ x^{k}}(S)\leqslant v(S)\}$, and ${ x^k}$ the optimal solution at
  step $k-1$.

For each \(i \in N\), we differentiate \(L(x,\lambda)\) with respect to \(x_i\):
\[
\frac{\partial L}{\partial x_i} = 2 \sum_{S \in \cA_k \,:\, i \in S} \gamma_S \left( x(S) - v(S) \right) - \lambda = 0.
\]
In addition, the primal feasibility condition requires:

\[ x(N) = v(N).\]
 From the first system of equations we deduce that $\sum_{S\in\cA_k:i\in
    S}\gamma_S(x(S)-v(S))$ does not depend on $i$, which can be written as
  \[
  \sum_{S\in\cA_k:1\in S}\gamma_S(x(S)-v(S))=\sum_{S\in\cA_k:i\in
    S}\gamma_S(x(S)-v(S)), \ \ i=2,\ldots,n
  \]
In summary, the whole set of equalities yields
\begin{equation}\label{eq:alg_lineal_sis}
\sum_{\stackrel{1 \in S}{S\in \cA_{k}}} \gamma_S x(S) -  \sum_{\stackrel{i \in S}{S\in \cA_{k}}} \gamma_S x(S)= \sum_{\stackrel{1 \in S}{S\in \cA_{k}}} \gamma_S v(S) -  \sum_{\stackrel{i \in S}{S\in \cA_{k}}} \gamma_S v(S) ,\forall i\in \{2, \ldots, n \}, \quad  x(N)=v(N).
\end{equation}

Hence, we obtain a linear system with $n$ equations and $n$ unknowns that can be solved efficiently.

\begin{example}
Consider the game $v$ given by

\begin{center}
\begin{tabular}{|c|ccccccc|}
\hline $S$ & 1 & 2 & 3 & 12 & 13 & 23 & 123 \\
\hline $v(S)$    & 37 & 7 & 92 & 35 & 64 & 19 & 77 \\
\hline
\end{tabular}
\end{center}

Let us apply the previous procedure in this case:

{\bf Step 0:} We consider all the subsets. Hence our equations are:

Comparing 1 and 2: $(v(1)-x_1) + (v(13)- x_1-x_3) = (v(2)-x_2) + (v(23)- x_2-x_3).$
Thus, we obtain the equation
$$ -2x_1 + 2x_2=-75.$$
Similarly, comparing 1 and 3 we obtain the equation
$$ -2x_1 + 2x_3=-39.$$
Hence, our system is given by

$$ \begin{array}{ccccccc}
-2x_1 & +&2x_2 &    &     & = & -75\\
-2x_1 &  &      & +& 2x_3 & = & -39 \\
x_1   & +& x_2 & +&x_3   & = & 77
\end{array}$$
The solution of this system is
$$ x_1=31.667, x_2=-5.833, x_3=51.167.$$
Then,

\begin{center}
\begin{tabular}{|c|ccccccc|}
\hline $S$ & 1 & 2 & 3 & 12 & 13 & 23 & 123 \\
\hline $v(S)$    & 37 & 7 & 92 & 35 & 63 & 19 & 77 \\
$x(S)$    & 31.667 & -5.833 & 51.167 & 25.833 & 82.833 & 45.333 & 77 \\
\hline
\end{tabular}
\end{center}

Hence, $\cA_1 = \{ 1,2,3,12\} .$

{\bf Step 1:} Now, our equations are

Comparing 1 and 2: $v(1)-x_1 = v(2)-x_2.$
Thus, we obtain the equation
$$ -x_1 + x_2=-30.$$
Comparing 1 and 3: $(v(1)-x_1) + (v(12)- x_1-x_2) = (v(3)-x_3).$
Thus, we obtain the equation
$$ -2x_1 - x_2 + x_3=20.$$
Hence, our system is given by
$$ \begin{array}{ccccccc}
-x_1 & +&x_2 &   &      & = & -30\\
-2x_1 & -&x_2   & +& x_3 & = & 20 \\
x_1   & +& x_2 & +&x_3   & = & 77
\end{array}$$
The solution of this system is

$$ x_1=23.4, x_2=-6.6, x_3=60.2.$$
Then,

\begin{center}
\begin{tabular}{|c|ccccccc|}
\hline $S$ & 1 & 2 & 3 & 12 & 13 & 23 & 123 \\
\hline $v(S)$    & 37 & 7 & 92 & 35 & 63 & 19 & 77 \\
$x(S)$    & 23.4 & -6.6 & 60.2 & 16.8 & 83.6 & 53.6 & 77 \\
\hline
\end{tabular}
\end{center}

Hence, $\cA_2 = \{ 1,2,3,12\} ,$ and as $\cA_1= \cA_2$ the procedure stops.
Consequently, applying Corollary \ref{corclosest}, $v^*$ is given by

\begin{center}
\begin{tabular}{|c|ccccccc|}
\hline $S$ & 1 & 2 & 3 & 12 & 13 & 23 & 123 \\
\hline $v(S)$    & 37 & 7 & 92 & 35 & 64 & 19 & 77 \\
 $x(S)$    & 23.4 & -6.6 & 60.2 & 16.8 & 83.6 & 53.6 & 77 \\
 $v^*(S) $ & 23.4 & -6.6 & 60.2 & 16.8 & 64 & 19 & 77 \\
\hline
\end{tabular}
\end{center}
Note that, the same as for Example \ref{Example1}, $x$ is the only imputation in $C(v^*).$
\end{example}

\begin{remark}
It is important to note that the objective function at $k$-th iteration $\phi_k(x)$ is convex (being a sum of squared terms), but it is not necessarily strictly convex. This can occur when the Hessian is only positive semidefinite, which means there are directions along which the function is flat. Consequently, it may happen that several $x$ minimize $\phi_k(x).$ To select a unique solution among the multiple optimal ones, an additional criterion may be imposed, such as choosing the solution with minimal norm or one that satisfies another desired property. In our case, we have considered the Moore-Penrose pseudoinverse. This pseudoinverse, denoted by $ A^+ ,$ is defined using the singular value decomposition (SVD). If

\[ A = U \Sigma V^T,\]
where $U \in \mathbb{R}^{m \times m} $ and $V \in \mathbb{R}^{n \times n} $ are orthogonal matrices, and $\Sigma \in \mathbb{R}^{m \times n}$ is a diagonal matrix with the singular values $ \sigma_1, \dots, \sigma_r $ (and zeros elsewhere), then we define

\[ A^+ := V \Sigma^+ U^T,\]
where $\Sigma^+$ is the diagonal matrix obtained by replacing each nonzero singular value $\sigma_i$ with its reciprocal $1/\sigma_i.$

If we denote the associated system at iteration $k$ in CLOBIS algorithm by $M_{k-1} x_{k}=b_k$,
we will solve it using the Moore-Penrose pseudoinverse:
$$x_{k}=M_{k-1}^{+}b_k.$$
Proceeding this way, we always obtain a unique solution to the problem associated to iteration $k$.
\end{remark}

Let us show that when this algorithm terminates, then it yields the closest balanced game $v^*.$

\begin{theorem}\label{convergence}
Let $\cA^*$ be a family of subsets of $N$ such that  an optimal solution $x^*$ to the problem
\[ \min_x \sum_{S \in \cA^*} \gamma_S \left(v(S) - x(S)\right)^2, \quad \text{subject to} \quad x(N) = v(N), \]
satisfies the conditions
\[ x^*(S) \leqslant v(S), \quad \forall S \in \cA^*, \quad \text{and} \quad x^*(S) > v(S), \quad \forall S \notin \cA^*. \]
Then, $x^*$ is an optimal solution to Problem~(\ref{eq:alg_general}).
\end{theorem}

\begin{proof}
The first step is to prove the convexity of the function
\[ g(x) = \sum_{S \subseteq N} \gamma_S \left( v(S) - x(S) \right)^2 {\bm 1}_{x(S)\leqslant v(S)} .\]
For this, observe that for $S\subseteq N,$
\[ \left( v(S) - x(S) \right)^2 {\bm 1}_{x(S) \leqslant v(S) }= \left[ \max\left\{ v(S) - x(S), 0 \right\} \right]^2. \]
Thus, we can express $g(x)$ as:
\[ g(x) = \sum_{S \subseteq N} \gamma_S \left[ \max\left\{ v(S) - x(S), 0 \right\} \right]^2. \]
To prove convexity, let us observe that $v(S)-x(S)=:h(x)$ is an affine function,
  hence in particular convex. Now, by composition of convex functions,
  $\max(h(x),0)$ is also convex, and also nonnegative.  Next, the square function $u \mapsto u^2$ is convex and non-decreasing on $[0,\infty)$, preserving convexity when composed with a nonnegative convex function. Finally, $g(x)$ is convex since each term is convex in $x$, and the sum of convex functions with nonnegative weights $\gamma_S$ is convex.

We now show that if the family of subsets $\cA^*$ and $x^*$ satisfy the
properties described above, then $x^*$ is a local minimum of $g.$ For this, consider a partition for all subsets $S\subseteq N$ into three classes: $\cA^*_1 = \{\,S\in\cA^* : x^*(S) < v(S)\}$, $\cA^*_2 = \{\,S\in\cA^* : x^*(S) = v(S)\}$, and $\cA^*_3 = \{\,S\notin\cA^* : x^*(S) > v(S)\}.$ Since the inequalities in $\cA^*_1$ and $\cA^*_3$ are strict at $x^*$, there exists a neighborhood $\mathcal{N}$ of $x^*$ such that for every $x'\in\mathcal{N}$: $x'(S) < v(S),\ \forall S\in\cA^*_1,$ and $x'(S) > v(S),\ \forall S\in\cA^*_3.$
Define $g_i(x) \;=\;\sum_{S\in \cA^*_i} \gamma_S (v(S)-x(S))^2\,\bm{1}_{x(S)\le v(S)} \quad (i=1,2,3),$
so that $g(x)=g_1(x)+g_2(x)+g_3(x)$.

On $\cA^*_3$, since $x'(S)>v(S)$ and $x^*(S)>v(S)$, both indicator functions vanish, hence $g_3(x') \;=\; g_3(x^*) \;=\;0.$ On $\cA^*_2$, we have $g_2(x^*)=0$ (since $x^*(S)=v(S)$), while for any $x'$,
  \[
    g_2(x') \;=\;\sum_{S\in\cA^*_2}\gamma_S (v(S)-x'(S))^2\,\bm{1}_{x'(S)\le v(S)}
    \;\ge\; 0.
  \]

On $\cA^*_1$, we get

\[
    g_1(x^*) \;=\;\sum_{S\in\cA^*_1}\gamma_S(v(S)-x^*(S))^2, \ \ g_1(x') \;=\;\sum_{S\in\cA^*_1}\gamma_S(v(S)-x'(S))^2.
  \]

By hypothesis, $x^*$ is an optimal solution to the problem associated to the family of subsets $\cA^*=\cA^*_1 \cup \cA^*_2.$ Since $x^*(S)=v(S)$ for all $S \in \cA^*_2,$ let us see that $x^*$ is also an optimal solution to the problem associated to the family of subsets $\cA^*_1$ and therefore $g_1(x^*)\leq g_1(x').$ Indeed, $x^*$ satisfies both the system of equations (see Eq. (\ref{eq:alg_lineal_sis})) that determine the optimal solution for $\cA^*$ and the one corresponding to $\cA^*_1.$ That is, $x^*(N)=v(N),$ and for $i=2,\ldots,n$:

  \[
  \sum_{S\in\cA^*:1\in S}\gamma_S(x^*(S)-v(S))=\sum_{S\in\cA^*:i\in
    S}\gamma_S(x^*(S)-v(S)) \Leftrightarrow
  \]

  \[
  \sum_{S\in\cA^*_1:1\in S}\gamma_S(x^*(S)-v(S))+ \sum_{S\in\cA^*_2:1\in S}\gamma_S(x^*(S)-v(S))=\sum_{S\in\cA^*_1:i\in
    S}\gamma_S(x^*(S)-v(S))+ \sum_{S\in\cA^*_2:i\in
    S}\gamma_S(x^*(S)-v(S))
  \]

  \[
\Leftrightarrow  \sum_{S\in\cA^*_1:1\in S}\gamma_S(x^*(S)-v(S))=\sum_{S\in\cA^*_1:i\in
    S}\gamma_S(x^*(S)-v(S)).
  \]

Thus, $x^*$ is optimal for $g_1.$ Combining these inequalities,
\[
g(x')
= g_1(x') + g_2(x') + g_3(x')
\;\ge\; g_1(x^*) + 0 + 0
= g(x^*),
\qquad
\forall\,x'\in\mathcal{B}.
\]
Hence $x^*$ is a local minimum of $g$. Finally, since $g$ is convex over the convex set $\{x : x(N) = v(N) \}$,
every local minimum is global, and the theorem follows.

\end{proof}

  \begin{corollary}\label{cor:optim}
    Let $v\in\cG(n)$.
Supposing that the CLOBIS algorithm terminates, the final output $x^*$ is a core
element of $v^*$, unique solution of Problem~(\ref{eq:optimization_problem2}), given by
$v^*(S)=\min(x^*(S),v(S))$, $S\subseteq N$.
\end{corollary}
  \begin{proof}
  Suppose the algorithm terminates with $x^*$ and $\cA^*$. This solution satisfies
  conditions of Theorem~\ref{convergence}, therefore, $x^*$ is an optimal
  solution of Problem~(\ref{eq:alg_general}). By Proposition~\ref{prop:equiv}, it follows
  that $(v^*,x^*)$ is an optimal solution of
  Problem~(\ref{eq:optimization_problem2}), where $v^*$ is given by (\ref{closest-game}) in Corollary~\ref{corclosest}. Moreover, as this problem is strictly
  convex, $v^*$ is unique.
\end{proof}


Note however that it could be the case that the algorithm gets into a cycle, as next example shows:
\begin{example}
Consider the game $v$ given by

\begin{center}
\begin{tabular}{|c|*{15}{c}|}
\hline Subset & 1 & 2 & 3 & 4 & 12 & 13 & 14 & 23 & 24 & 34 & 123 & 124 & 134 & 234 & 1234\\
\hline $v$    & $-48$ & $-32$ & 25 & $-74$ &  12 & 100 & $-60$ &  62 & $-35$ &  32  & 54 &  29 &  21 & $-75$ &  57 \\
\hline
\end{tabular}
\end{center}

Let us apply the algorithm to this game:

\textbf{Iteration 1:}

Linear system:
\[
\begin{bmatrix}
4 & -4 & 0 & 0 \\
0 & 4 & -4 & 0 \\
0 & 0 & 4 & -4 \\
1 & 1 & 1 & 1
\end{bmatrix}
x
=
\begin{bmatrix}
93 \\ -204 \\ 381 \\ 57
\end{bmatrix}.
\]

Solution:
\[
x^\star = (30,\; 6.75,\; 57.75,\; -37.5).
\]

\begin{center}
\begin{tabular}{|c|*{15}{c}|}
\hline $S$ & 1 & 2 & 3 & 4 & 12 & 13 & 14 & 23 & 24 & 34 & 123 & 124 & 134 & 234 & 1234 \\
\hline $v(S)$ & $-48$ & $-32$ & 25 & $-74$ & 12 & 100 & $-60$ & 62 & $-35$ & 32 & 54 & 29 & 21 & $-75$ & 57 \\
\hline $x(S)$ & 30 & 6.7 & 57.7 & $-37.5$ & 36.7 & 87.7 & $-7.5$ & 64.5 & $-30.7$ & 20.2 & 94.5 & $-0.7$ & 50.2 & 27 & 57 \\
\hline
\end{tabular}
\end{center}

Active coalitions:
\[
\mathcal{A}_1 = \{13, 34, 124 \}.
\]

\textbf{Iteration 2:}

Linear system:
\[
\begin{bmatrix}
1 & 0 & 1 & 0 \\
0 & 1 & -2 & 0 \\
0 & -1 & 1 & -1 \\
1 & 1 & 1 & 1
\end{bmatrix}
x
=
\begin{bmatrix}
100 \\ -103 \\ 71 \\ 57
\end{bmatrix}.
\]

Solution:
\[
x^\star = (72,\; -47,\; 28,\; 4).
\]

\begin{center}
\begin{tabular}{|c|*{15}{c}|}
\hline $S$ & 1 & 2 & 3 & 4 & 12 & 13 & 14 & 23 & 24 & 34 & 123 & 124 & 134 & 234 & 1234 \\
\hline $v(S)$ & $-48$ & $-32$ & 25 & $-74$ & 12 & 100 & $-60$ & 62 & $-35$ & 32 & 54 & 29 & 21 & $-75$ & 57 \\
\hline $x(S)$ & 72 & $-47$ & 28 & 4 & 25 & 100 & 76 & $-19$ & $-43$ & 32 & 53 & 29 & 104 & $-15$ & 57 \\
\hline
\end{tabular}
\end{center}

Active coalitions:
\[
\mathcal{A}_2 = \{2,13,23,24,34,123,124\}.
\]

\textbf{Iteration 3:}

Linear system:
\[
\begin{bmatrix}
1 & -3 & 0 & -1 \\
0 & 3 & -2 & 1 \\
1 & 0 & 3 & -2 \\
1 & 1 & 1 & 1
\end{bmatrix}
x
=
\begin{bmatrix}
105 \\ -170 \\ 222 \\ 57
\end{bmatrix}.
\]

Solution:
\[
x^\star = (37.875,\; -17.40625,\; 51.4375,\; -14.90625).
\]

\begin{center}
\begin{tabular}{|c|*{15}{c}|}
\hline $S$ & 1 & 2 & 3 & 4 & 12 & 13 & 14 & 23 & 24 & 34 & 123 & 124 & 134 & 234 & 1234 \\
\hline $v(S)$ & $-48$ & $-32$ & 25 & $-74$ & 12 & 100 & $-60$ & 62 & $-35$ & 32 & 54 & 29 & 21 & $-75$ & 57 \\
\hline $x(S)$ & 37.9 & $-17.4$ & 51.4 & $-14.9$ & 20.5 & 89.3 & 23.0 & 34.0 & $-32.3$ & 36.5 & 71.9 & 5.6 & 74.4 & 19.1 & 57 \\
\hline
\end{tabular}
\end{center}

Active coalitions:
\[
\mathcal{A}_3 = \{13,23,124\}.
\]

\textbf{Iteration 4:}

Solution:
\[
x^\star = (72,\; 34,\; 28,\; -77).
\]

\begin{center}
\begin{tabular}{|c|*{15}{c}|}
\hline $S$ & 1 & 2 & 3 & 4 & 12 & 13 & 14 & 23 & 24 & 34 & 123 & 124 & 134 & 234 & 1234 \\
\hline $v(S)$ & $-48$ & $-32$ & 25 & $-74$ & 12 & 100 & $-60$ & 62 & $-35$ & 32 & 54 & 29 & 21 & $-75$ & 57 \\
\hline $x(S)$ & 72 & 34 & 28 & $-77$ & 106 & 100 & $-5$ & 62 & $-43$ & $-49$ & 134 & 29 & 23 & $-15$ & 57 \\
\hline
\end{tabular}
\end{center}

Active coalitions:
\[
\mathcal{A}_4 = \{4,13,23,24,34,124\}.
\]

\textbf{Iteration 5:}

Solution:
\[
x^\star = (42.429,\; 8.571,\; 53.714,\; -47.714).
\]

\begin{center}
\begin{tabular}{|c|*{15}{c}|}
\hline $S$ & 1 & 2 & 3 & 4 & 12 & 13 & 14 & 23 & 24 & 34 & 123 & 124 & 134 & 234 & 1234 \\
\hline $v(S)$ & $-48$ & $-32$ & 25 & $-74$ & 12 & 100 & $-60$ & 62 & $-35$ & 32 & 54 & 29 & 21 & $-75$ & 57 \\
\hline $x(S)$ & 42.4 & 8.5 & 53.7 & $-47.7$ & 51 & 96.1 & $-5.3$ & 62.3 & $-39.1$ & 6 & 104.7 & 3.3 & 48.4 & 14.6 & 57 \\
\hline
\end{tabular}
\end{center}

Active coalitions:
\[
\mathcal{A}_5 = \{13,24,34,124\}.
\]

\textbf{Iteration 6 (Cycle Detected):}

Solution:
\[
x^\star = (68,\; -43,\; 28,\; 4).
\]

\begin{center}
\begin{tabular}{|c|*{15}{c}|}
\hline $S$ & 1 & 2 & 3 & 4 & 12 & 13 & 14 & 23 & 24 & 34 & 123 & 124 & 134 & 234 & 1234 \\
\hline $v(S)$ & -48 & -32 & 25 & -74 & 12 & 100 & -60 & 62 & -35 & 32 & 54 & 29 & 21 & -75 & 57 \\
\hline $x(S)$ & 68 & -43 & 28 & 4 & 25 & 96 & 72 & -15 & -39 & 32 & 53 & 29 & 100 & -11 & 57 \\
\hline
\end{tabular}
\end{center}

Active coalitions:
\[
\mathcal{A}_6 = \{2,13,23,24,34,123,124\} =\mathcal{A}_2.
\]

This set of active coalitions is identical to the one obtained in Iteration 2, and the algorithm gets stuck in a cycle.

\end{example}

\begin{remark}
At this point, it should be remarked that this situation almost never occurs when generating $v$ at random.  It has been computationally verified that the proportion of cases in which cycles occur decreases rapidly as the number of players increases. In fact, for $n \geqslant 6$, it becomes extremely difficult to find examples where such cycles arise.

Although cycles rarely occur, it is straightforward to adapt the CLOBIS
  algorithm to prevent cycling. Note that the initial family in the algorithm
  $\mathcal{A}_0$, is the family of all subsets of $N$, though this choice can
  be modified. Thus, if starting with this family $\mathcal{A}_0$ leads to a
  cycle, one can restart with a different family (generated at random) and
  repeat the process. By doing this repeatedly, one will eventually reach a
  situation where the initial family is optimal, and the algorithm converges. In
  practice, changing the initial family only a few times is usually sufficient
  to achieve convergence.\footnote{Our implementation in R incorporates this improvement.}
\end{remark}

Thanks to this efficient version of the algorithm we have gone from working with a problem of $2^n+n+n-3$ variables to one with $n$ variables. This makes the calculation of $v^*$ much faster and more efficient. In this way we can calculate $v^*$ for larger games (larger than 20 players). In Table \ref{tab:alg_time} we can see the average time the algorithm takes for each $n$ and also the average number of iterations when $v$ is generated uniformly in the cube $[-L,L]^{2^{n}-1}$ for $L=10$ (compare with Table \ref{tab:execution_times}).

\begin{table}[htb]
\centering
\begin{tabular}{ccr}
\toprule
$n$ & Mean Iter & Mean Time(s) \\
\midrule
3  & 2.7 & 0.000972 \\
4  & 2.9 & 0.000878 \\
5  & 3.5 & 0.001037 \\
6  & 3.0 & 0.002062 \\
7  & 3.6 & 0.002622 \\
8  & 3.4 & 0.002634 \\
9  & 3.3 & 0.004450 \\
10 & 3.6 & 0.008055 \\
11 & 3.8 & 0.012719 \\
12 & 3.8 & 0.025711 \\
13 & 3.9 & 0.049686 \\
14 & 3.9 & 0.118343 \\
15 & 3.8 & 0.276413 \\
16 & 4.0 & 0.599793 \\
17 & 4.2 & 1.511491 \\
18 & 4.1 & 3.401944 \\
19 & 4.0 & 5.726339\\
20 & 4.1 & 14.140730 \\
\bottomrule
\end{tabular}\caption{Average number of iteration and mean time of execution for Algorithm \ref{alg: iterative}}
\label{tab:alg_time}
\end{table}

We can see that the number of iterations is very small even for large referential sets $N$.
We finish this section by giving a relation between $\cA^*$, the family
  given in the last iteration of the algorithm, and the facet in which $v^*$ lies.
\begin{lemma}
  Let $v\in \cG(n)$ and $\cB\in \gB^*(n).$ The closest balanced game $v^*$ to $v$ is in the facet defined by $\cB$ if and only if $\cB \subseteq \cA^*.$
\end{lemma}

\begin{proof}
 Recall that $v^*$ is in the facet defined by $\cB$ with balancing
  weights $(\lambda_S)_{S\in\cB}$ if and only if
$$ \sum_{S\in \cB} \lambda_S v^*(S)= v^*(N).$$

Consider a m.b.c. $\cB$ and suppose that $\forall S\in {\cB},$ then
$S\in \cA^*.$  Denoting by $x^*$ the output of the algorithm, we know by
  Corollary~\ref{cor:optim} that $x^*$ is a core element of $v^*$, and they are
  in relation with $v$ by (\ref{closest-game}). It follows that for any $S\in\cB$, then
  $S\in\cA^*$, and $x^*(S)=v^*(S)$. Therefore,
$$ \sum_{S\in \cB} \lambda_S v^*(S) = \sum_{S\in \cB} \lambda_S x^*(S)=  x^*(N)= v^*(N).$$

Conversely, suppose there exists $S\in \cB$ s.t. $S\not\in \cA^*.$ In this case, $x^*(S)>v(S)= v^*(S).$ Then,

$$ \sum_{S\in \cB} \lambda_S v^*(S)< \sum_{S\in \cB} \lambda_S x^*(S) = x^*(N)= v^*(N).$$
\end{proof}

\section{Simulation study}\label{sec:simu}

 Consider a game $v$ with an empty core and its projection $v^*$ on
$\cBG_{v(N)}(n)$. An important question is to know whether the core of $v^*$ is
  reduced to a point or not, which depends on which
facet or face of  $\cBG_{v(N)}(n)$ the game $v^*$
lies. Theorems~\ref{th:pcbg} and \ref{th:pcbg1} give a precise answer to this question.

The goal of this section is to estimate by simulation the probability that $v^*$
has a singleton core, as well as the probability that $v^*$ lies in a facet. We will consider throughout the distance with uniform weights
$\gamma_S=1$.\footnote{It has been checked experimentally that the results are not
sensitive to the weights.}

\subsection{Probability to have a singleton core}\label{sec:prsi}
  In this section, we conduct experiments to
estimate the probability that the projection $v^*$ of a game $v$ on
$\cBG_{v(N)}$ has a singleton core. For this, we will randomly generate a large
number of games and find their projections. To this aim, we randomly generate
values $v(S)$, $S\subset N$, uniformly in the interval $[-L, L]$, where $L$
takes values 1, 10 and 100. All values $v(S)$ are identically distributed and
drawn independently. To always work in the same polyhedron, we set the value of $v(N) = 0.$
For each value of $L$,  we have generated 5,000 games
(not necessarily unbalanced) to
estimate the probability of having a singleton core. The results of this simulation study appear in Table~\ref{tab:4}.
\begin{table}[h]
\centering
\begin{tabular}{|c|c|c|c|}
\hline
$n/L$ & 1 & 10 & 100 \\
\hline
3 & 0.668 & 0.669 & 0.688 \\
4 & 0.946 & 0.941 & 0.946 \\
5 & 0.999 & 0.999 & 0.999 \\
6 & 1.000 & 1.000 & 1.000 \\
7 & 1.000 & 1.000 & 1.000 \\
\hline
\end{tabular}
\caption{Probability that $v^*$ has a single core for different values of $n$ and $L$.}
\label{tab:4}
\end{table}

From Table \ref{tab:4}, we can see that as $n$ increases, the probability that the core
of $v^*$ is a singleton tends to 1 very quickly. Moreover, we observe that
the value of $L$ does not seem to affect this conclusion. We have also run simulations for other values of $v(N)$, as well as cases where the projection is made on $\cBG(n)$ (without the constraint $v^*(N)=v(N)$), or where both the games and their projections are restricted to nonnegative values. The conclusions remain essentially the same.

In Section~\ref{sec:prob}, we will provide a mathematical proof of
these facts.

%
%

\subsection{Probability of facets and faces}\label{sec:fafa}
In this subsection, we estimate through simulation the probability that a game
with an empty core is projected on a facet or on a face of lower dimension
(i.e., an intersection of facets) of the set of balanced games\footnote{When we
say that a game is projected on a face, we mean the face of lowest
dimension containing the projected game.}.  As it can be seen from
Figure~\ref{fig:mosaico}, even for a very simple polyhedron, the probability to
be projected on a face which is not a facet is far to be negligible.

We have conducted a detailed study of the cases $n=3$ and $n=4$,
  proceeding as above to generate random values of $v(S)$, $S\subset N$,
  uniformly in the interval\footnote{No notable difference
  is observed with different values of $L$.} $[-10,10].$ For $n=3$, there are 5
  m.b.c., which are: $\{1,23\}$, $\{1,2,3\}$, $\{12,13,23\}$, $\{2,13\}$ and
  $\{3,12\}$, and we estimate the probability to fall on a given face of
  $\cBG_{v(N)}(n)$ ( with $v(N)=0$) by
  generating 20 000 games and counting how many games fall on that face. The
  results are reported in Table~\ref{tab:8}. Faces are coded as follows: each
face is displayed as a string of zeros and ones, where the ones indicate the
facets that form the face. For example, the face ``10000'' corresponds to the
first facet associated with the first balanced collection (in the order as listed above),
while ``11000'' represents the intersection of the first two facets.

\begin{table}[h]
\centering
\begin{tabular}{|c|c|c|}
\hline
Face & Proportion & Single Core \\
\hline
00000 & $0.096$ & N \\
00001 & $0.075$ & N \\
10000 & $0.073$ & N \\
00010 & $0.072$ & N \\
11111 & $0.048$ & Y \\
01000 & $0.045$ & Y \\
00011 & $0.045$ & Y \\
10010 & $0.044$ & Y \\
00100 & $0.042$ & Y \\
10001 & $0.042$ & Y \\
01010 & $0.040$ & Y \\
00110 & $0.037$ & Y \\
01001 & $0.036$ & Y \\
11000 & $0.036$ & Y \\
10100 & $0.035$ & Y \\
00101 & $0.034$ & Y \\
00111 & $0.033$ & Y \\
11001 & $0.033$ & Y \\
01011 & $0.032$ & Y \\
11010 & $0.031$ & Y \\
10110 & $0.031$ & Y \\
10101 & $0.030$ & Y \\
\hline
\end{tabular}
\caption{Probability for each face that $v^*$ lies in it (listed in decreasing value)}
\label{tab:8}
\end{table}

Several observations can be made from Table \ref{tab:8}. First, the most probable face in
Table \ref{tab:8} is the face defined as ``00000''. This face represents the
case when $v$ is a balanced game. Thus, this is an estimate of the probability
that a randomly generated game is balanced. Another observation is that the
lineality space formed by the intersection of all facets ``11111'' has a relatively
high probability around $0.05$. Note that the faces  with games having a
non-singleton core are those obtained as intersections of facets such that $\vert \cB\vert <
n$ (collections 1, 4 and 5). If we sum the probabilities associated with the
case of a singleton core, we obtain $0.68$ (compare this with Table
\ref{tab:4}), hence the probability to have a
nonsingleton core is $0.32.$ If we restrict to the case
where $v$ is not balanced, this last proportion decreases to $0.22$.
 As it seems natural, most often the probability to fall on some facet is greater than
the probability to fall on some face of lower dimension, and curiously, it is
more probable to fall on a facet yielding a non-singleton core than on a facet
yielding a singleton core. The total probability to fall on a facet (resp.,  a face
of lower dimension, except the lineality space) is 0.30 (resp. 0.60).

%

For $n=4$ we observe similar results. In this case, we have 41 m.b.c. and hence
the number of different faces is very high, so that we do not list them here. We
have conducted the same experiment as above but generating 50 000 games. The
combinations leading to the lineality space (the intersection of all facets) appear
with a probability of $0.0003$, and the one associated with all zeros
(the interior of the polytope of balanced games) has a probability of
$0.002.$ In this case, the probability to fall on faces with a
singleton core increases to $0.95$ (compare again with Table \ref{tab:4}). The
probability to fall on a facet is
0.02, which is much lower than for $n=3$. This means that most of the time the
projected game will belong to a face of lower dimension. This motivates the
theoretical study of Section~\ref{sec:fasi}, where the proportion of faces with
a singleton core is studied.

\section{On the probability that $v^*$ has a singleton core}\label{sec:prob}

As illustrated by the simulations in Section~\ref{sec:simu},
projections of randomly generated games often have a singleton core.
We first study the proportion of facets and then of arbitrary faces
whose games all have a singleton core, using the full cone $\cBG(n)$.
The affine correspondence established in Section~\ref{sec:baco}
transfers these combinatorial results to $\cBG_\alpha(n)$.
Finally, we examine under which probabilistic condition they imply
an asymptotic statement for the weighted closest balanced game with
$v^*(N)=v(N)$.

\subsection{Proportion of facets whose games have a singleton core}
\label{sec:prfacet}

The goal of this section is to formally prove that the proportion of facets such
that all games inside the facet have a singleton core tends to 1 when $n$
increases. We already know from Th. \ref{th:bgfacet} that facets are in bijection with minimal balanced
collections (excluding the trivial m.b.c. $\{ N\}$). Let $B_n=|\gB^*(n)|$ be
the number of m.b.c.  of a set of $n$ elements, excluding the trivial
m.b.c. ($B_3=5$, $B_4=41$, etc.; see Table \ref{tablembmc} which counts all m.b.c., including $\{N\}$). As explained in Section~\ref{sec:baco}, given a m.b.c. $\cB,$ it follows that $|\cB| \in \{ 2,\ldots, n\} .$ Now, for a fixed set of $n$ elements, let us denote by $B_n^m$ the number of minimal balanced collections with $m$ subsets i.e.
$$ B_n^m:= | \{ \cB\in \gB^*(n) : |\cB|=m\} .$$
Hence,
$$B_n=\sum_{m=2}^n B_n^m.$$
Moreover, by Theorem~\ref{th:pcbg}, given $\cB\in\gB^*(n)$ and its corresponding facet $\cF_\cB$, every game in $\cF_\cB$ has a singleton core if and only if $|\cB|=n$. Hence, the proportion of facets containing at least one game with a non-singleton core is
$$
\phi_n:=\frac{B_n-B_n^n}{B_n}.
$$
Our aim is to prove that
$$
\lim_{n\to\infty}\phi_n=0.
$$
As explained in Section~\ref{sec:baco}, after fixing an ordering of the
coalitions in $\cB,$ one can associate to each $\cB \in \gB^*(n)$ a 0/1-valued
matrix $M_{\cB}.$ If $|\cB |=n,$ the corresponding matrix $M_{\cB}$ is a square
matrix $n\times n.$ As the weights for a m.b.c. are unique and positive, it
follows that $M_{\cB}$ is invertible and $M_{\cB}^{-1} \bm 1 >\bm 0$, with $\bm
0$ denoting the $n$-dim zero vector, and $>$ indicates coordinatewise strict inequality.

Consider $\gC_n$ the set of $-1/1$-valued matrices of size $n\times n.$ Similarly, let us denote by $\gM_n$ the set of 0/1-valued matrices of size $n\times n.$ Let us denote by $\gC_n'$ (resp. $\gM_n'$) the subset of $-1/1$-valued matrices (resp. 0/1-valued matrices) being regular.

Suppose that a matrix \(C\in\gC_n\) is generated by choosing its entries independently, with each entry taking the values \(-1\) and \(1\) with probability \(1/2\). Similarly, suppose that a matrix \(M\in\gM_n\) is generated by choosing its entries independently, with each entry taking the values \(0\) and \(1\) with probability \(1/2\). This induces the uniform probability distribution \(\mathbb{P}_n\) on \(\gC_n\) (resp. \(\gM_n\)), so that each matrix has probability \(2^{-n^2}\). We use the same notation \(\mathbb{P}_n\) for both distributions, since the underlying space will always be clear from the context.

We consider the following auxiliary result.

\begin{lemma}\cite{tik20}\label{lem:regular1}
For the sequence of probability spaces $(\gC_n,\mathbb{P}_n)$,
$n\in\mathbb{N}$, it follows that
\[
\mathbb{P}_n\{C\notin\gC_n'\}
=
\left(\frac{1}{2}+o_n(1)\right)^n.
\]
Similarly, for the sequence of probability spaces
$(\gM_n,\mathbb{P}_n)$, $n\in\mathbb{N}$, we have
\[
\mathbb{P}_n\{M\notin\gM_n'\}
=
\left(\frac{1}{2}+o_n(1)\right)^n.
\]
In particular,
\[
\lim_{n\to\infty}\mathbb{P}_n\{\gC_n'\}=1
\qquad\text{and}\qquad
\lim_{n\to\infty}\mathbb{P}_n\{\gM_n'\}=1.
\]
\end{lemma}

Both estimates follow from Theorem A in \cite{tik20}. In that theorem, $B_n(p)$ denotes an $n\times n$ matrix with independent Bernoulli$(p)$ entries taking values in $\{0,1\}$, and the theorem is stated for
\[
B_n(p)+s\mathbf 1\mathbf 1^\top, \qquad s\in[-1,0].
\]

For $\gM_n$,
one takes $p=1/2$ and $s=0$. For $\gC_n$, one takes $p=1/2$ and
$s=-1/2$, observing that $2(B_n(1/2)-\tfrac12\mathbf 1\mathbf 1^\top)$ has i.i.d. $\{-1,1\}$ entries and that multiplication by $2$ does not affect singularity.

We use the standard asymptotic notation. For two sequences $f(n)$ and
$g(n)>0$, we write $f(n)=o(g(n))$ if
\[
\frac{f(n)}{g(n)}\longrightarrow 0
\qquad\text{as } n\to\infty,
\]
and $f(n)=O(g(n))$ if there exists a constant $C>0$ such that
$|f(n)|\leq Cg(n)$ for all sufficiently large $n$. In particular,
$o_n(1)$ denotes a quantity tending to zero as $n\to\infty$.

Let us define the function

\begin{equation}\label{function-f}
\begin{array}{cccc} f:& \gC_n & \rightarrow & \gM_n \\
                         & C     & \mapsto      & \frac{C+{\bf 1}{\bf 1}^\top }{2} \end{array}
\end{equation}
which  maps the entries \(-1\) and \(1\) of \(C\) to \(0\) and \(1\), respectively. Moreover, $f$ is a bijection between $\gC_n$ and $\gM_n.$
Let us now define
$$ \begin{array}{cccc} v:& \gC_n' & \rightarrow & \mathbb{R}^n \\
                         & C     & \mapsto      & C^{-1} {\bf 1}\end{array}$$

\begin{lemma}\label{lemma-aux2}
Let $C\in\gC_n$ be selected according to $\mathbb{P}_n.$ Then
\[
\lim_{n\to\infty}
\mathbb{P}_n
\left\{
v(C)_i\neq 0,\quad i=1,\ldots,n
\ \middle|\
C\in\gC_n'
\right\}
=1.
\]
\end{lemma}

\begin{proof}
Let $C\in\gC_n$ be selected according to $\mathbb{P}_n.$ Defining
$q_n:=\mathbb{P}_n\{C\notin\gC_n'\}$, we get by Lemma~\ref{lem:regular1},
\[
q_n=\left(\frac{1}{2}+o_n(1)\right)^n.
\]
Fix $i\in\{1,\ldots,n\}$. If $C\in\gC_n'$, we have by Cramer's rule
\[
v(C)_i=
\frac{\det(H_{C,i}^{\bf 1})}{\det(C)},
\]
where $H_{C,i}^{\bf 1}$ is the matrix obtained from $C$ by replacing
its $i$-th column by ${\bf 1}$. Hence, under the condition $C\in\gC_n'$,  $v(C)_i=0$ if and only if $H_{C,i}^{\bf
  1}$ is singular.
We first show that
\begin{equation}\label{eq:fixed-column-sing}
\mathbb{P}_n\{H_{C,i}^{\bf 1}\text{ is singular}\}=q_n.
\end{equation}
Denote by $C_{\cdot i}$ the $i$-th column of $C$. For
${\bf w}\in\{-1,1\}^n$, consider the conditional probability
\[
\mathbb{P}_n
\left\{
C\text{ is singular}
\ \middle|\
C_{\cdot i}={\bf w}
\right\}.
\]
We claim that this probability does not depend on ${\bf w}$.
Indeed, take ${\bf w}_1,{\bf w}_2\in\{-1,1\}^n$ and define
\[
D:=
\operatorname{diag}
\left(
w_{2,1}w_{1,1},\ldots,w_{2,n}w_{1,n}
\right).
\]
Then $D{\bf w}_1={\bf w}_2$, and the mapping $C\longmapsto DC$
is a bijection between the set of matrices whose $i$-th column is
${\bf w}_1$ and the set of matrices whose $i$-th column is
${\bf w}_2$. Moreover, since $D$ is invertible, $C$ is singular if and only if
$DC$ is singular. Therefore,
\[
\mathbb{P}_n
\left\{
C\text{ is singular}
\ \middle|\
C_{\cdot i}={\bf w}_1
\right\}
=
\mathbb{P}_n
\left\{
C\text{ is singular}
\ \middle|\
C_{\cdot i}={\bf w}_2
\right\}.
\]
Since $\mathbb{P}_n\{C_{\cdot i}={\bf w}\}=2^{-n}$
for every ${\bf w}\in\{-1,1\}^n$, the law of total probability gives
\[
q_n
=
\mathbb{P}_n
\left\{
C\text{ is singular}
\ \middle|\
C_{\cdot i}={\bf 1}
\right\}.
\]
On the other hand, $H_{C,i}^{\bf 1}$ is uniformly distributed over
the set of $-1/1$-valued matrices whose $i$-th column is ${\bf 1}$.
Thus, Eq.~\eqref{eq:fixed-column-sing} follows.

Consequently, for every $i\in\{1,\ldots,n\}$,
\begin{eqnarray*}
\mathbb{P}_n
\left\{
v(C)_i=0
\ \middle|\
C\in\gC_n'
\right\}
&=&
\frac{
\mathbb{P}_n
\left\{
v(C)_i=0,\ C\in\gC_n'
\right\}
}{
\mathbb{P}_n\{C\in\gC_n'\}
}
\\
&\leqslant&
\frac{
\mathbb{P}_n
\left\{
H_{C,i}^{\bf 1}\text{ is singular}
\right\}
}{
\mathbb{P}_n\{C\in\gC_n'\}
}
\\
&=&
\frac{q_n}{1-q_n}.
\end{eqnarray*}
Notice that this bound does not depend on $i$. Hence, by the union
bound,
\begin{eqnarray*}
\mathbb{P}_n
\left\{
\exists i\in\{1,\ldots,n\}:v(C)_i=0
\ \middle|\
C\in\gC_n'
\right\}
&\leqslant&
\sum_{i=1}^n
\mathbb{P}_n
\left\{
v(C)_i=0
\ \middle|\
C\in\gC_n'
\right\}
\\
&\leqslant&
\frac{nq_n}{1-q_n}.
\end{eqnarray*}
Now fix any $\varepsilon\in(0,1/2)$. Since
\[
q_n=\left(\frac{1}{2}+o_n(1)\right)^n,
\]
for $n$ sufficiently large,
\[
q_n\leqslant
\left(\frac{1}{2}+\varepsilon\right)^n.
\]
Therefore,
\[
0\leqslant
\frac{nq_n}{1-q_n}
\leqslant
\frac{
n\left(\frac{1}{2}+\varepsilon\right)^n
}{
1-\left(\frac{1}{2}+\varepsilon\right)^n
}
\longrightarrow 0.
\]
Hence,
\[
\lim_{n\to\infty}
\mathbb{P}_n
\left\{
v(C)_i\neq0,\quad i=1,\ldots,n
\ \middle|\
C\in\gC_n'
\right\}
=1,
\]
as desired.
\end{proof}

We adopt from now on the notation $f(n)\sim g(n)$ whenever
$\lim_{n\to\infty}f(n)/g(n)=1$.
Let
\[
\gM'_{n,+}:=
\left\{
M\in\gM_n':M^{-1}{\bf 1}>\bm 0
\right\}.
\]
For the next result, we will use the Sherman--Morrison formula \cite{shmo50}: if
$A$ is regular and $1+\mathbf{v}^{\top}A^{-1}\mathbf{u}\neq0$, then
\[
(A+\mathbf{u}\mathbf{v}^{\top})^{-1}
=
A^{-1}
-
\frac{
A^{-1}\mathbf{u}\mathbf{v}^{\top}A^{-1}
}{
1+\mathbf{v}^{\top}A^{-1}\mathbf{u}
}.
\]

\begin{lemma}\label{lem:regular3}
It follows that
\[
\mathbb{P}_n\{\gM_{n,+}'\}
\sim
\frac{1}{2^{n-1}}.
\]
\end{lemma}

\begin{proof}
Let
\[
\gD_n:=
\left\{
D=\operatorname{diag}(d_1,\ldots,d_n):
d_i\in\{-1,1\},\ i=1,\ldots,n
\right\}.
\]
The set $\gD_n$ acts on $\gC_n$ by multiplication on the right,
$C\mapsto CD$. For $C\in\gC_n$, denote its orbit by
\[
\mathcal{O}_{\gC_n}(C)
:=
\{CD:D\in\gD_n\}.
\]
Thus, two matrices belong to the same orbit if and only if one can be
obtained from the other by changing the signs of some of its columns.

These orbits form a partition of $\gC_n$, and every orbit has exactly
$2^n$ elements. Indeed, if
$\mathcal{O}_{\gC_n}(C)\cap\mathcal{O}_{\gC_n}(C')\neq\emptyset$, then
$CD=C'D'$ for some $D,D'\in\gD_n$, and hence
$C'=CDD'^{-1}\in\mathcal{O}_{\gC_n}(C)$, which implies
$\mathcal{O}_{\gC_n}(C')=\mathcal{O}_{\gC_n}(C)$. Finally, if $CD=CD'$ and $D\neq D'$, then some column
of $C$ would be equal to its opposite, which is impossible for a
$-1/1$-valued matrix.

Through the bijection $f$ in Eq.~\eqref{function-f}, each orbit induces
the following family in $\gM_n$:
\[
\mathcal{O}_{\gM_n}(C)
:=
\left\{
M_D(C):=
f(CD)
=
\frac{CD+\mathbf{1}\mathbf{1}^{\top}}{2}
:
D\in\gD_n
\right\}.
\]
In terms of $0/1$-valued matrices, changing the sign of a column of
$C$ amounts to complementing the corresponding column of $M_D(C)$.
Hence $\gM_n$ is also partitioned into orbits of cardinality $2^n$.
Since
\[
|\gM_n|=2^{n^2},
\]
the number of such orbits is $2^{n^2-n}$.

We now show that every orbit contains at most two matrices in
$\gM_{n,+}'$, and that asymptotically almost every orbit contains
exactly two.

\medskip
\noindent
\emph{Case 1: $C$ is regular.}

Let $\mathbf{v}:=C^{-1}\mathbf{1}$.
Assume that, for some $D\in\gD_n$, $M_D(C)\in\gM_{n,+}'$,
and put
\[
\mathbf{x}:=M_D(C)^{-1}\mathbf{1}>\mathbf{0}.
\]
Since $M_D(C)\mathbf{x}=\mathbf{1}$,
\[
(CD+\mathbf{1}\mathbf{1}^{\top})\mathbf{x}
=
2\cdot\mathbf{1},
\]
and therefore
\[
CD\mathbf{x}
=
(2-\mathbf{1}^{\top}\mathbf{x})\mathbf{1}.
\]
Multiplying by $C^{-1}$ and using $D^{-1}=D$, we obtain
\[
\mathbf{x}
=
(2-\mathbf{1}^{\top}\mathbf{x})D\mathbf{v}.
\]
Since $\mathbf{x}>\mathbf{0}$, the scalar
$2-\mathbf{1}^{\top}\mathbf{x}$ is nonzero and all coordinates of
$D\mathbf{v}$ must have the same sign.

Suppose first that $\mathbf{v}$ has no zero coordinate. There are then
only two possible choices of $D$ for which all coordinates of
$D\mathbf{v}$ have the same sign:
\[
D_+
:=
\operatorname{diag}
\bigl(
\operatorname{sgn}(v_1),\ldots,\operatorname{sgn}(v_n)
\bigr),
\qquad
D_-:=-D_+.
\]
Indeed,
\[
D_+\mathbf{v}=|\mathbf{v}|,
\qquad
D_-\mathbf{v}=-|\mathbf{v}|,
\]
where $|\mathbf{v}|$ denotes the vector of coordinatewise absolute
values. Thus the orbit can contain at most two elements of
$\gM_{n,+}'$.

We next show that both choices actually give matrices in
$\gM_{n,+}'$. Put
\[
a:=\|\mathbf{v}\|_1
=
\sum_{j=1}^n |v_j|.
\]
Since $C\mathbf{v}=\mathbf{1}$, every row $\mathbf{c}_i^{\top}$ of
$C$ satisfies $\mathbf{c}_i^{\top}\mathbf{v}=1$. Hence
\[
1
=
|\mathbf{c}_i^{\top}\mathbf{v}|
\leq
\|\mathbf{v}\|_1
=
a.
\]
In fact, $a>1$ for $n\geq2$. Otherwise $a=1$, and equality in the
triangle inequality, together with $v_j\neq0$ for every $j$, would
imply $c_{ij}=\operatorname{sgn}(v_j)$ for every $i,j$.
Thus all rows of $C$ would be identical, contradicting the regularity
of $C$. It follows that
\[
1+\mathbf{1}^{\top}D_+\mathbf{v}
=
1+a>0,
\qquad
1+\mathbf{1}^{\top}D_-\mathbf{v}
=
1-a<0.
\]
Using $(CD)^{-1}=DC^{-1}$, Eq.~\eqref{function-f} and applying the Sherman--Morrison formula  with
$A=CD$ and $u=v=\bm 1$ gives, after simplification
\[
M_D(C)^{-1}\mathbf{1}
=
\frac{
2D\mathbf{v}
}{
1+\mathbf{1}^{\top}D\mathbf{v}
}.
\]
Consequently,
\[
M_{D_+}(C)^{-1}\mathbf{1}
=
\frac{2|\mathbf{v}|}{1+a}
>
\mathbf{0},
\]
whereas
\[
M_{D_-}(C)^{-1}\mathbf{1}
=
\frac{2|\mathbf{v}|}{a-1}
>
\mathbf{0}.
\]
Therefore,
\[
M_{D_+}(C),\ M_{D_-}(C)\in\gM_{n,+}',
\]
and the orbit contains exactly two elements of $\gM_{n,+}'$.

If $\mathbf{v}$ has a zero coordinate, then every vector
$D\mathbf{v}$ has a zero coordinate. Hence it can be neither strictly
positive nor strictly negative, and the preceding relation shows that
the orbit contains no element of $\gM_{n,+}'$.

\medskip
\noindent
\emph{Case 2: $C$ is singular.}

We show that also in this case the orbit contains at most two elements
of $\gM_{n,+}'$. If $\operatorname{rank}(C)\leq n-2$, then, for every $D\in\gD_n$,
\[
\operatorname{rank}
(CD+\mathbf{1}\mathbf{1}^{\top})
\leq
\operatorname{rank}(C)+1
\leq n-1.
\]

Hence every $M_D(C)$ is singular, and the orbit contains no element of
$\gM_{n,+}'$. Suppose now that $\operatorname{rank}(C)=n-1$ and put
$\ker(C)=\operatorname{span}\{\mathbf{z}\}$.
Assume that $M_D(C)\in\gM_{n,+}'$ for some $D\in\gD_n$, and put
\[
\mathbf{x}:=M_D(C)^{-1}\mathbf{1}>\mathbf{0}.
\]
Let $\mathbf{u}\neq\mathbf{0}$ satisfy $\mathbf{u}^{\top}C=\mathbf{0}^{\top}$.
Since $CD+\mathbf{1}\mathbf{1}^{\top}$ is regular, necessarily $\mathbf{u}^{\top}\mathbf{1}\neq0$;
otherwise we would have $\mathbf{u}^{\top}
(CD+\mathbf{1}\mathbf{1}^{\top})
=
\mathbf{0}^{\top}$, a contradiction.
From
\[
CD\mathbf{x}
+
(\mathbf{1}^{\top}\mathbf{x})\mathbf{1}
=
2\cdot\mathbf{1},
\]
left multiplication by $\mathbf{u}^{\top}$ yields
\[
(\mathbf{1}^{\top}\mathbf{x})
\mathbf{u}^{\top}\mathbf{1}
=
2\mathbf{u}^{\top}\mathbf{1}.
\]
Hence, $\mathbf{1}^{\top}\mathbf{x}=2$, and consequently $CD\mathbf{x}=\mathbf{0}$.
Thus, $D\mathbf{x}\in\ker(C)=
\operatorname{span}\{\mathbf{z}\}$, so that $D\mathbf{x}=\alpha\mathbf{z}$
for some $\alpha\neq0$.

Since $\mathbf{x}>\mathbf{0}$, the vector $D\mathbf{x}$ has no zero
coordinate. Therefore, if $\mathbf{z}$ has a zero coordinate, no such
$D$ can exist. Otherwise, the signs of the diagonal entries of $D$ are
determined by the signs of the coordinates of $\mathbf{z}$, up to
reversing all signs simultaneously. Hence there are at most two
possible choices of $D$.

We conclude that every orbit in $\gM_n$ contains at most two elements
of $\gM_{n,+}'$.

Finally, Lemmas~\ref{lem:regular1} and~\ref{lemma-aux2} imply
\begin{align*}
&\mathbb{P}_n
\left\{
C\in\gC_n':
(C^{-1}\mathbf{1})_i\neq0,\ i=1,\ldots,n
\right\}\\
&\quad=
\mathbb{P}_n
\left\{
(C^{-1}\mathbf{1})_i\neq0,\ i=1,\ldots,n
\,\middle|\,
C\in\gC_n'
\right\}
\mathbb{P}_n\{C\in\gC_n'\}
\longrightarrow1.
\end{align*}
This property is constant on each orbit: if $D\in\gD_n$, then $CD$ is
regular if and only if $C$ is regular, and
\[
(CD)^{-1}\mathbf{1}
=
DC^{-1}\mathbf{1},
\]
so the presence or absence of zero coordinates is unchanged.

Since all orbits have the same cardinality, a proportion $1-o_n(1)$ of
the $2^{n^2-n}$ orbits has this property. Each of these orbits contains
exactly two elements of $\gM_{n,+}'$, while every orbit contains at
most two. Therefore,
\[
2(1-o_n(1))2^{n^2-n}
\leq
|\gM_{n,+}'|
\leq
2\cdot 2^{n^2-n}.
\]
It follows that
\[
|\gM_{n,+}'|
\sim
2^{n^2-n+1}.
\]
Since $|\gM_n|=2^{n^2}$,
\[
\mathbb{P}_n\{\gM_{n,+}'\}
=
\frac{|\gM_{n,+}'|}{|\gM_n|}
\sim
\frac{1}{2^{n-1}},
\]
as desired.
\end{proof}

%
%

Since $|\gM_n|=2^{n^2}$, Lemma~\ref{lem:regular3} is equivalently
written as
\begin{equation}\label{num-mbc}
|\gM'_{n,+}|
\sim
2^{n^2-n+1}.
\end{equation}

There is in fact an exact correspondence between matrices in
$\gM'_{n,+}$ and ordered minimal balanced collections consisting of
$n$ coalitions.
Indeed, let $\cB$ be a minimal balanced collection consisting of
$n$ coalitions, and fix an ordering of its coalitions. Its incidence
matrix $M_{\cB}$ is regular and, by Eq.~\eqref{eq:mbc},
\[
M_{\cB}^{-1}{\bf 1}=\bm\lambda>\bm0.
\]
Hence $M_{\cB}\in\gM'_{n,+}$.
Conversely, let $M\in\gM'_{n,+}$ and put
$\bm\lambda:=M^{-1}{\bf 1}>\bm0$. The columns of $M$ are incidence
vectors of coalitions and
\[
M\bm\lambda={\bf 1},
\]
so they form a balanced collection. This collection is minimal.
Indeed, if a proper subcollection were balanced, its balancing
weights could be extended by zeros to a vector $\bm\mu\in\mathbb R^n$
satisfying $M\bm\mu={\bf 1}$. Since $M$ is regular,
$\bm\mu=\bm\lambda$, which is impossible because
$\bm\lambda>\bm0$, whereas $\bm\mu$ has at least one zero
coordinate.

Therefore, matrices in $\gM'_{n,+}$ are in bijection with ordered
minimal balanced collections consisting of $n$ coalitions. Since
the columns of a regular matrix are distinct, each minimal balanced
collection gives exactly $n!$ different orderings. Consequently,
\begin{equation}\label{number-single}
B_n^n=\frac{|\gM'_{n,+}|}{n!}.
\end{equation}
Combining Eqs.~\eqref{num-mbc} and~\eqref{number-single} yields the
following result.

\begin{theorem}\label{theo2}
Let $N$ be a set of $n$ players, and let $B_n^n$ denote the number
of minimal balanced collections on $N$ consisting of $n$ coalitions.
Then
\[
B_n^n
\sim
\frac{2^{n^2-n+1}}{n!}.
\]
\end{theorem}

Next step is to study the asymptotic growth of the whole set of m.b.c. $B_n$. We will prove that the asymptotically dominant term of $B_n$ is $B^n_n$ and that the terms associated with smaller supports $m<n$ grow at a slower rate.

\begin{theorem}\label{theo3}
$B_n$ has the same asymptotic growth as $B^n_n$:
$$B_n \sim \dfrac{2^{n^2-n+1}}{n!}.$$
\end{theorem}

\begin{proof}
We first note that
$$ \frac{B_n}{B_n^n} = 1+\sum_{m=2}^{n-1}\frac{B_n^m}{B_n^n}.$$
The idea of the proof is to study the asymptotic behavior on $n$ of ${B_n^m\over B_n^n}$ for $m<n$. Recall from Theorem~\ref{theo2} and Eq. (\ref{number-single}) that
\[ B_n^n=\frac{|\gM'_{n,+}|}{n!}\sim \frac{2^{n^2-n+1}}{n!}.\]
Let us replace the asymptotic equivalence $|\gM'_{r,+}| \sim 2^{r^2-r+1}$ for
some integer $r\geq 2$ with
uniform bounds.
By Eq. (\ref{num-mbc}), we know that $ \lim_{n \to \infty}
\frac{|\gM'_{n,+}|}{2^{n^2-n+1}} = 1.$ Hence, for the ball of
  radius 1/2 centered at 1, there exists an integer $n_0\geq 3$ such that for all $n \geq n_0$,

\begin{equation}\label{eq:uniform-lower-M}
\frac{1}{2} \leq \frac{|\gM'_{n,+}|}{2^{n^2-n+1}} \leq \frac{3}{2},
\end{equation}
which directly provides the lower bound for $n\geq n_0$ given by

$$ B_n^n\geq {{1\over 2} 2^{n^2-n+1} \over n!}.$$
Furthermore, to obtain a global upper bound for the numerator that remains valid even for small values of $n,$ where the asymptotic behavior has not yet stabilized (for $n<n_0$), we define the constant:

\[ A := \max\left\{ \frac{3}{2}, \max_{2 \leq n < n_0} \frac{|\gM'_{n,+}|}{2^{n^2-n+1}} \right\},\]
that is well-defined as the second maximum is taken over a finite set. This choice of $A$ ensures that, for every $n \geq 2$,

\begin{equation}\label{eq:uniform-upper-M}
|\gM'_{n,+}| \leq A \, 2^{n^2-n+1}.
\end{equation}

We now look for an upper bound of $B_n^m$ for $2\leq m<n$. Let $\cB$ be a
minimal balanced collection on $N$ with $m$ coalitions, and let $M_{\cB}$ be its
$n\times m$ incidence matrix. Since $\cB$ is a m.b.c., as explained in Section
2, $\operatorname{rank}(M_{\cB})=m$, and the equation
\[ M_{\cB}\bm\lambda=\bm 1 \]
has a unique solution satisfying $\bm\lambda>\bm0$ (see Eq. (\ref{eq:mbc})).
Since $M_{\cB}$ has rank $m$, there exists an index set $I\subseteq N$ with $|I|=m$ such that the submatrix $M_{\cB}[I]$ formed by the rows indexed by $I$ is regular. Moreover,

\[ M_{\cB}[I]\bm\lambda=\bm 1,\]
and therefore
\[M_{\cB}[I]^{-1}\bm1=\bm\lambda>\bm0,\]
so that $M_{\cB}[I]\in\gM'_{m,+}$.
There are $\binom{n}{m}$ possible choices for the set $I,$ although some of them may fail to form a matrix $M_{\cB}[I]$ in the previous conditions. For a fixed valid $I$ and the corresponding regular matrix $M_{\cB}[I]$, each of the remaining $n-m$ rows of $M_{\cB}$ is a non-zero vector and there are $2^m-1$ possible vectors in these conditions.  Finally, each minimal balanced collection with $m$ coalitions is represented by exactly $m!$ matrices obtained by permuting the columns. A given minimal balanced collection may be counted for more than one
choice of $I$, which is harmless since we are deriving an upper bound. Therefore,

\begin{equation}\label{eq:bound-Bnm}
B_n^m \leq \binom{n}{m} \frac{|\gM'_{m,+}|(2^m-1)^{n-m}}{m!}.
\end{equation}
Combining \eqref{eq:bound-Bnm} with
\eqref{eq:uniform-upper-M}, \eqref{eq:uniform-lower-M}, and using the inequality
$(2^m-1)^{n-m}\leq 2^{m(n-m)}$, we obtain, for every $n\geq n_0$ and every $2\leq m\leq n-1$,

\begin{align}
\frac{B_n^m}{B_n^n}
&\leq
\binom{n}{m}\frac{n!}{m!}
\frac{
|\gM'_{m,+}|(2^m-1)^{n-m}
}{
2^{n^2-n}
}
\nonumber\\
&\leq
A\binom{n}{m}\frac{n!}{m!}
\,2^{m^2-m+1+m(n-m)-(n^2-n)}
\nonumber\\
&=
2A\binom{n}{m}\frac{n!}{m!}
\,2^{m(n-1)-(n^2-n)}.
\label{eq:ratio-Bnm-Bnn}
\end{align}
Now, for a fixed value $n$, define the numbers

\[ a_{n,m}:= \binom{n}{m}\frac{n!}{m!} \,2^{m(n-1)-(n^2-n)}, \qquad 2\leq m\leq n-1. \]
Let us show that $a_{n,m}$ is increasing in $m.$ For $2\leq m\leq n-2,$ we have

\[ \frac{a_{n,m+1}}{a_{n,m}} = \frac{2^{n-1}(n-m)}{(m+1)^2} \geq \frac{2^n}{(n-1)^2}>1. \]
Consequently,
\[\max_{2\leq m\leq n-1}a_{n,m} = a_{n,n-1}.\]
A direct computation gives

\[ a_{n,n-1} = \binom{n}{n-1}\frac{n!}{(n-1)!}\,2^{(n-1)^2-(n^2-n)} = \frac{n^2}{2^{n-1}}.\]
It follows from \eqref{eq:ratio-Bnm-Bnn} that

\begin{equation}\label{eq:uniform-ratio}
\max_{2\leq m\leq n-1}\frac{B_n^m}{B_n^n} \leq 2A\,\frac{n^2}{2^{n-1}}.
\end{equation}
Finally,
\begin{align*}
\frac{B_n}{B_n^n} & = 1+\sum_{m=2}^{n-1}\frac{B_n^m}{B_n^n} \\
& \leq 1+(n-2)\max_{2\leq m\leq n-1}\frac{B_n^m}{B_n^n} \\
& \leq 1+2A(n-2)\frac{n^2}{2^{n-1}}.
\end{align*}
Since $\frac{n^3}{2^n}\longrightarrow 0$, we conclude that $\frac{B_n}{B_n^n}\longrightarrow 1$.
Therefore,
\[ B_n\sim B_n^n\sim \frac{2^{n^2-n+1}}{n!}, \]
as claimed.
\end{proof}

\begin{remark}
Note that the previous Theorem provides an asymptotic value for $B_n,$ the
number of minimal balanced collections. This is a particularly useful asymptotic result, since, to the best of our knowledge, exact values of this sequence are currently known only up to $n=7$. The problem of
enumerating m.b.c.'s has been studied first by Peleg in \cite{pel65}, where an
inductive algorithm for generating the m.b.c.'s for a referential set of $n$
elements via the m.b.c.'s for $n-1$ elements is established. In that paper the
values up to $n=4$ are obtained. Recently in \cite{lagrsu23}, this algorithm has
been implemented, and values
up to $n=7$ have been obtained. As it can be seen in Table \ref{tablembmc}, the values of $B_n$
increase very fast when $n$ grows.
\end{remark}

Combining Theorems~\ref{theo2} and~\ref{theo3}, we obtain the desired result as a corollary.

\begin{corollary}\label{cor2}
The proportion of facets containing games with a non-singleton core tends to 0 when $n$ tends to infinity:

\[ \lim_{n \to \infty} \phi_n = 0.\]
\end{corollary}



\subsection{Proportion of faces whose games have a singleton core}\label{sec:fasi}

The results in the previous subsection show that the proportion of facets
associated with minimal balanced collections of cardinality smaller than $n$
tends to zero as $n$ grows. In this section, we address the analogous problem
for arbitrary faces of $\cBG(n)$.

Recall from Theorem~\ref{th:pcbg1} that if a face is obtained as the
intersection of the facets associated with minimal balanced collections
$\cB_1,\ldots,\cB_p$, then every game in that face has a singleton core if
and only if
\[
\operatorname{rank}
\bigl\{
\mathbf{1}^S:S\in\cB_1\cup\cdots\cup\cB_p
\bigr\}
=n.
\]
In particular, if a face is contained in a facet associated with a minimal
balanced collection of cardinality $n$, then every game in the face has a
singleton core. However, this condition is only sufficient: the rank may
equal $n$ even when none of the minimal balanced collections defining the
face has cardinality $n$.

Notice also that intersections of facets of $\cBG(n)$ are always nonempty.
Indeed, the lineality space $L=\Lin(\cBG(n))$ is contained in every facet
of $\cBG(n)$. Therefore, for any family of facets
$\{\mathcal H_{\cB}:\cB\in I\}$,
\[
L\subseteq \bigcap_{\cB\in I}\mathcal H_{\cB}.
\]
Thus, every intersection of facets considered below is a nonempty face of
$\cBG(n)$.

Before counting faces, we must take into account that the same face can be
represented as an intersection of facets in several different ways. We
therefore introduce a canonical maximal representation.

Let $I\subseteq\gB^*(n)$ and write
\[
\mathscr F_I
=
\bigcap_{\cB\in I}\mathcal H_{\cB},
\qquad
U(I):=
\bigcup_{\cB\in I}\cB,
\]
where $\mathcal H_{\cB}$ denotes the facet associated with the minimal
balanced collection $\cB$.

\begin{lemma}\label{lem:maximal-facet-representation}
Let $I\subseteq\gB^*(n)$ and let
\[
\mathscr F_I=\bigcap_{\cB\in I}\mathcal H_{\cB}.
\]
For $\cC\in\gB^*(n)$, the facet $\mathcal H_{\cC}$ contains
$\mathscr F_I$ if and only if
\[
\cC\subseteq U(I).
\]
\end{lemma}

\begin{proof}
Let $L$ denote the lineality space of $\cBG(n)$. Every facet of
$\cBG(n)$ contains $L$, and therefore so does every face.

By Theorem~\ref{th:4}, the extreme rays of $\cBG(n)$ outside $L$ are
indexed by the nonempty proper coalitions of $N$. More precisely, write
\[
\rho_S=-\delta_S,
\qquad |S|>1,
\]
and
\[
\rho_{\{i\}}=r_i,
\qquad i\in N.
\]
A direct inspection of the balancedness inequality~\eqref{bal} gives,
for every $\cB\in\gB^*(n)$,
\[
\rho_S\in\mathcal H_{\cB}
\quad\Longleftrightarrow\quad
S\notin\cB,
\qquad
\varnothing\neq S\subsetneq N.
\]
Since a face of a polyhedral cone is determined by its lineality space
together with the extreme rays it contains, the face $\mathscr F_I$ is
entirely determined by
\[
\{\rho_S:S\notin U(I)\}
=
\{\rho_S:
\rho_S\in\mathcal H_{\cB} \text{ for every } \cB\in I\}.
\]
Consequently, for $\cC\in\gB^*(n)$,
\begin{align*}
\mathscr F_I\subseteq\mathcal H_{\cC}
&\quad\Longleftrightarrow\quad
[\forall S, S\notin U(I)\Longrightarrow S\notin\cC]\\
&\quad\Longleftrightarrow\quad
\cC\subseteq U(I).
\end{align*}
This proves the claim.
\end{proof}

\begin{definition}
For $I\subseteq\gB^*(n)$, we define the \textbf{maximal representation}
of the face $\mathscr F_I$ by
\[
\operatorname{cl}_{\gB}(I)
:=
\left\{
\cC\in\gB^*(n):
\cC\subseteq U(I)
\right\}.
\]
\end{definition}

By Lemma~\ref{lem:maximal-facet-representation},
$\operatorname{cl}_{\gB}(I)$ is precisely the family of all minimal balanced collections corresponding to facets
containing $\mathscr F_I$, and is therefore uniquely determined by the
face. Moreover, the proof shows that the union $U(I)$ itself determines
the face: for two families $I,J\subseteq\gB^*(n)$,
\[
\mathscr F_I=\mathscr F_J
\quad\Longleftrightarrow\quad
U(I)=U(J).
\]

We illustrate this representation with a simple example.

\begin{example}\label{ex:rep_example}
Consider $N=\{1,2,3\}$ and let $I=\{\cB_1,\cB_2,\cB_3,\cB_4\}$,
where
\begin{align*}
\cB_1&=\{\{1\},\{2,3\}\},\\
\cB_2&=\{\{2\},\{1,3\}\},\\
\cB_3&=\{\{3\},\{1,2\}\},\\
\cB_4&=\{\{1,2\},\{1,3\},\{2,3\}\}.
\end{align*}
Then $U(I)=\{
\{1\},\{2\},\{3\},
\{1,2\},\{1,3\},\{2,3\}
\}$, which contains every nonempty proper coalition of $N$.
Hence the maximal representation contains every nontrivial minimal
balanced collection. In particular, it also contains
\[
\cB_5=\{\{1\},\{2\},\{3\}\},
\]
which does not belong to $I$. Thus $\operatorname{cl}_{\gB}(I)=\{\cB_1,\ldots,\cB_5\}$.
Since $U(I)$ contains all nonempty proper coalitions, the corresponding
face contains no extreme ray outside the lineality space. Therefore,
$\mathscr F_I$ coincides with the lineality space of $\cBG(3)$.
\end{example}

We now compare the number of faces whose games all have a singleton core
with the number of faces containing at least one game with a
non-singleton core.

Let $\mathscr F_n$ denote the set of all faces of $\cBG(n)$, and define
\[
\mathscr F_n^n
:=
\left\{
\mathscr F\in\mathscr F_n:
\text{every game in }\mathscr F
\text{ has a singleton core}
\right\},
\]
and
\[
\mathscr F_n^{<n}
:=
\mathscr F_n\setminus\mathscr F_n^n.
\]
Thus, $\mathscr F_n^{<n}$ is the family of faces containing at least one
game with a non-singleton core.

The proof proceeds in three steps. First, we obtain an upper bound for
$|\mathscr F_n^{<n}|$. Second, we construct a sufficiently large family
of faces in $\mathscr F_n^n$, which provides a lower bound for
$|\mathscr F_n^n|$. Finally, we compare the two bounds.

\begin{proposition}\label{prop:bad-faces-bound}
The number of faces of $\cBG(n)$ containing at least one game with a
non-singleton core satisfies
\[
|\mathscr F_n^{<n}|
\leq
2^{\,2^{n-1}+n^2}.
\]
\end{proposition}

\begin{proof}
Let $\mathscr F\in\mathscr F_n^{<n}$ and let $U_{\mathscr F}$ denote
the union of all coalitions appearing in its maximal facet
representation. By Theorem~\ref{th:pcbg1},
\[
\operatorname{rank}
\bigl\{
\mathbf{1}^S:S\in U_{\mathscr F}
\bigr\}
<n.
\]
Hence the Boolean incidence vectors $\{\mathbf{1}^S:S\in U_{\mathscr F}\}$
are contained in a proper linear subspace
$V\subsetneq\mathbb R^n$ generated by vectors from $\{0,1\}^n$.
We first bound the number of such subspaces. Since $V$ has dimension at
most $n-1$, it admits a basis consisting of at most $n-1$ linearly
independent vectors from $\{0,1\}^n$. As there are $2^n$ different vectors from $\{ 0, 1\}^n,$ the number of
possible subspaces of this type is bounded by $\displaystyle\sum_{r=0}^{n-1}\binom{2^n}{r}$.
Using the crude bound
\[
\binom{2^n}{r}\leq (2^n)^r,
\]
we obtain
\[
\sum_{r=0}^{n-1}\binom{2^n}{r}
\leq
n\,2^{n(n-1)}
\leq
2^{n^2}.
\]

Next, fix a proper subspace $V\subsetneq\mathbb R^n$. We claim that $|V\cap\{0,1\}^n|\leq2^{n-1}$.
Indeed, since $V$ is proper, there exists a nonzero vector
$\bm a=(a_1,\ldots,a_n)\in\mathbb R^n$ such that
\[
V\subseteq
\ker(\bm a)
=
\{\bm x\in\mathbb R^n:\bm a\cdot \bm x=0\}.
\]
Choose $i$ with $a_i\neq0$. Once the coordinates
$\{x_j\}_{j\neq i}$ of a vector $\bm x\in\ker(\bm a)\cap\{0,1\}^n$ are fixed,
the equation
\[
a_ix_i
=
-\sum_{j\neq i}a_jx_j
\]
determines $x_i$ uniquely. Therefore, for each of the $2^{n-1}$ choices of the remaining coordinates, there is at most one possible value of $x_i\in\{0,1\}$. Hence $|\ker(\bm a)\cap\{0,1\}^n|\leq 2^{n-1},$
and therefore
\[
|V\cap\{0,1\}^n|\leq2^{n-1}.
\]
For a fixed subspace $V$, the family $U_{\mathscr F}$ is a subset of
$V\cap\{0,1\}^n$. Thus, there are at most
\[
2^{|V\cap\{0,1\}^n|}
\leq
2^{2^{n-1}}
\]
possible choices for $U_{\mathscr F}$.

Finally, by Lemma~\ref{lem:maximal-facet-representation}, a face is
uniquely determined by its maximal union of coalitions. Combining the
two estimates gives
\[
|\mathscr F_n^{<n}|
\leq
2^{n^2}\,2^{2^{n-1}}
=
2^{\,2^{n-1}+n^2}.
\]
\end{proof}

We now derive a lower bound for the number of faces in $\mathscr F_n^n$.
Let $R:=\{2,\ldots,n-1\}$.
Recall that $B_{n-2}^{\,n-2}$ denotes the number of minimal balanced
collections on $R$ consisting of exactly $n-2$ coalitions. Let $\cC$ be one such minimal balanced collection and define
\[ \cB := \cC\cup\bigl\{\{1\},\{n\}\bigr\}.
\]
Observe that $\cB$ is a minimal balanced
  collection on $N$. Indeed, it is balanced since by construction every $i\in N$
  satisfies (\ref{eq:bc}), and it is minimal because $\cC$ is minimal and
  $\{1\},\{n\}$ cannot be removed.
As it consists of $n$ coalitions, by Theorem~\ref{th:pcbg}, every game in
the facet associated with $\cB$ has a singleton core.

We now use these collections to construct many different faces.
Let $\cC_1,\cC_2$ be two distinct minimal balanced collections on $R$,
each consisting of $n-2$ coalitions, and put
\[
\cB_i
:=
\cC_i\cup\bigl\{\{1\},\{n\}\bigr\},
\qquad i=1,2.
\]
Define
\begin{equation}\label{def-U}
U_{\cC_1,\cC_2}
:=
\cB_1\cup\cB_2
=
(\cC_1\cup\cC_2)
\cup
\bigl\{\{1\},\{n\}\bigr\}.
\end{equation}
The family $U_{\cC_1,\cC_2}$ is complement-free, in the sense that $S\in U_{\cC_1,\cC_2}
\quad\Longrightarrow\quad
S^c\notin U_{\cC_1,\cC_2}$.
Indeed, every coalition in $U_{\cC_1,\cC_2}$ is either one of the
singletons $\{1\},\{n\}$ or a subset of $R=\{2,\ldots,n-1\}$.
If $S=\{1\}$, then $S^c=N\setminus\{1\}$ contains $n$ and has more
than one element, so it is neither $\{n\}$ nor a subset of $R$.
Hence $S^c\notin U_{\cC_1,\cC_2}$. The case $S=\{n\}$ is symmetric.
Finally, if $S\in\cC_1\cup\cC_2$, then $S\subseteq R$, and therefore
$\{1,n\}\subseteq S^c$. Thus $S^c$ is neither contained in $R$ nor
equal to $\{1\}$ or $\{n\}$, and consequently
$S^c\notin U_{\cC_1,\cC_2}$.

\begin{lemma}\label{lem:number-U}
Let $\mathcal U_n$ denote the set of distinct families
$U_{\cC_1,\cC_2}$ of the above type. Then
\begin{equation}\label{eq:number-U}
|\mathcal U_n|
\geq
\frac{
\binom{B_{n-2}^{\,n-2}}{2}
}{
\binom{2n-4}{n-2}^{\,2}
}.
\end{equation}
\end{lemma}

\begin{proof}
There are $\displaystyle\binom{B_{n-2}^{\,n-2}}{2}$
unordered pairs $\{\cC_1,\cC_2\}$ of distinct minimal balanced
collections of cardinality $n-2$ on $R$.
Fix one possible union $V=\cC_1\cup\cC_2$.
Since $\cC_1$ and $\cC_2$ are distinct and both have cardinality
$n-2$,
\[
n-1
\leq
|V|
\leq
2n-4.
\]
Each $\cC_i$ is an $(n-2)$-element subfamily of $V$. Therefore, even if
we count ordered pairs, the number of pairs
$(\cC_1,\cC_2)$ that can give rise to a fixed $V$ is at most
\[
\binom{|V|}{n-2}^{\,2}
\leq
\binom{2n-4}{n-2}^{\,2}.
\]
It follows that the number of distinct unions $V$ is at least
\[
\frac{
\binom{B_{n-2}^{\,n-2}}{2}
}{
\binom{2n-4}{n-2}^{\,2}
}.
\]
Finally, the correspondence $V\longmapsto
V\cup\bigl\{\{1\},\{n\}\bigr\}$
is injective because every coalition in $V$ is contained in $R$.
Therefore the same lower bound holds for $|\mathcal U_n|$.
\end{proof}

Let $\mathfrak P_n$ denote the family of all nontrivial partitions of
$N$ into two blocks:
\[
\mathfrak P_n
:=
\bigl\{
\{S,S^c\}:
\varnothing\neq S\subsetneq N
\bigr\}.
\]
Since $S$ and $S^c$ determine the same partition, $|\mathfrak P_n|=2^{n-1}-1$.
For $U\in\mathcal U_n$, define the family of
\textit{admissible partitions}
\[
\mathfrak P(U)
:=
\bigl\{
\{S,S^c\}\in\mathfrak P_n:
S\notin U,\ S^c\notin U
\bigr\}.
\]
Since $U$ is complement-free, each coalition of $U$ excludes exactly
one partition from $\mathfrak P_n$, and distinct coalitions of $U$
exclude distinct partitions. Hence
\[
|\mathfrak P(U)|
=
|\mathfrak P_n|-|U|.
\]
Moreover, $|U|=2+|\cC_1\cup\cC_2|\leq 2n-2$.
Therefore
\begin{equation}\label{eq:admissible-partitions}
|\mathfrak P(U)|
\geq
2^{n-1}-2n+1.
\end{equation}
For every $J\subseteq\mathfrak P(U)$, define
\[
W(U,J)
:=
U
\cup
\bigcup_{\{S,S^c\}\in J}
\{S,S^c\}.
\]
The map $(U,J)\longmapsto W(U,J)$
is injective. Indeed, since $U$ is complement-free and the partitions
in $J$ contain no coalition belonging to $U$,
\begin{equation}\label{eq:recover-U}
U
=
\bigl\{
S\in W(U,J):
S^c\notin W(U,J)
\bigr\}.
\end{equation}
Thus $U$ can be recovered from $W(U,J)$. Once $U$ is known, $J$ is
recovered as
\[
J
=
\bigl\{
\{S,S^c\}\in\mathfrak P_n:
S,S^c\in W(U,J)\setminus U
\bigr\}.
\]

We now associate a face of $\cBG(n)$ with every pair $(U,J)$.
For each $U\in\mathcal U_n$, choose one pair
$\{\cC_1,\cC_2\}$ giving rise to $U$, and set
\[
\cB_i
=
\cC_i\cup\bigl\{\{1\},\{n\}\bigr\},
\qquad i=1,2.
\]
For $J\subseteq\mathfrak P(U)$, define
\[
\mathscr F_{U,J}
:=
\mathcal H_{\cB_1}
\cap
\mathcal H_{\cB_2}
\cap
\bigcap_{\mathcal P \in J}\mathcal H_{\mathcal P}.
\]
Every $\mathcal P\in\mathfrak P_n$ is itself a minimal balanced
collection, and hence determines a facet of $\cBG(n)$.
The union of the coalitions appearing in the above representation is
exactly $W(U,J)$. Moreover, $\cB_1$ consists of $n$ coalitions and its
incidence matrix is regular. Hence
\[
\operatorname{rank}
\bigl\{
\mathbf{1}^S:S\in W(U,J)
\bigr\}
=n.
\]
By Theorem~\ref{th:pcbg1}, every game in $\mathscr F_{U,J}$ has a
singleton core.
Furthermore, by Lemma~\ref{lem:maximal-facet-representation}, the union
of coalitions in any facet representation uniquely determines the
resulting face. Therefore,
\[
W(U,J)\neq W(U',J')
\quad\Longrightarrow\quad
\mathscr F_{U,J}\neq\mathscr F_{U',J'}.
\]

\begin{proposition}\label{prop:good-faces-bound}
The number of faces of $\cBG(n)$ all of whose games have a singleton
core satisfies
\begin{equation}\label{eq:good-faces-lower-bound}
|\mathscr F_n^n|
\geq
\frac{
\binom{B_{n-2}^{\,n-2}}{2}
}{
\binom{2n-4}{n-2}^{\,2}
}
\,
2^{\,2^{n-1}-2n+1}.
\end{equation}
\end{proposition}

\begin{proof}
For every $U\in\mathcal U_n$, there are $2^{|\mathfrak P(U)|}$
possible choices of $J\subseteq\mathfrak P(U)$.
The map $(U,J)\mapsto W(U,J)$ is injective, and different families
$W(U,J)$ determine different faces. Therefore,
using Lemma~\ref{lem:number-U} and
Eq.~\eqref{eq:admissible-partitions},
\begin{align*}
|\mathscr F_n^n|
&\geq
|\mathcal U_n|\,
2^{\,\min_{U\in\mathcal U_n}|\mathfrak P(U)|}\\
&\geq
\frac{
\binom{B_{n-2}^{\,n-2}}{2}
}{
\binom{2n-4}{n-2}^{\,2}
}
\,
2^{\,2^{n-1}-2n+1}.
\end{align*}
\end{proof}

We are now in position to prove the main result of this section. Define
\[
\varphi_n
:=
\frac{
|\mathscr F_n^{<n}|
}{
|\mathscr F_n|
},
\]
the proportion of faces of $\cBG(n)$ containing at least one game with
a non-singleton core.

\begin{theorem}\label{th:faces-singleton-core}
The proportion of faces of $\cBG(n)$ containing at least one game with
a non-singleton core tends to zero:
\[
\lim_{n\to\infty}\varphi_n=0.
\]
Equivalently,
\[
\frac{|\mathscr F_n^n|}{|\mathscr F_n|}
\longrightarrow1.
\]
\end{theorem}

\begin{proof}

Applying Theorem~\ref{theo2} with $n-2$ players gives
\[
B_{n-2}^{\,n-2}
\sim
\frac{
2^{(n-2)^2-(n-2)+1}
}{
(n-2)!
}
=
\frac{
2^{n^2-5n+7}
}{
(n-2)!
}.
\]
Therefore
\[
\log_2 B_{n-2}^{\,n-2}
=
n^2-5n+7
-
\log_2((n-2)!)
+
o(1).
\]

By Stirling's formula,
\[
\log_2((n-2)!)
=
(n-2)\log_2(n-2)+O(n),
\]
and consequently $\log_2 B_{n-2}^{\,n-2}=
n^2-n\log_2n+O(n)$.
Since $B_{n-2}^{\,n-2}\to\infty$,
\[
\binom{B_{n-2}^{\,n-2}}{2}
=
\frac{(B_{n-2}^{\,n-2})^2}{2}
\left(1-\frac{1}{B_{n-2}^{\,n-2}}\right) \sim \frac{(B_{n-2}^{\,n-2})^2}{2},
\]
and hence
\[
\log_2\binom{B_{n-2}^{\,n-2}}{2}
=
2\log_2 B_{n-2}^{\,n-2}-1+o(1)
=
2n^2-2n\log_2n+O(n).
\]
Moreover,
\[
\binom{2n-4}{n-2}
\leq \sum_{i=0}^{2n-4}\binom{2n-4}{i} = (1+1)^{2n-4}=
2^{2n-4},
\]
so that
\[
2\log_2
\binom{2n-4}{n-2}
\leq
4n-8.
\]
Applying Proposition~\ref{prop:good-faces-bound}, we conclude that
\begin{align}
\log_2|\mathscr F_n^n|
&\geq
\log_2\binom{B_{n-2}^{\,n-2}}{2}
-
2\log_2\binom{2n-4}{n-2}
+
2^{n-1}-2n+1
\nonumber\\
&\geq
2^{n-1}
+
2n^2
-
2n\log_2n
+
O(n).
\label{eq:good-faces-asymptotic-lower}
\end{align}
On the other hand, Proposition~\ref{prop:bad-faces-bound} gives
\[
\log_2|\mathscr F_n^{<n}|
\leq
2^{n-1}+n^2.
\]
Consequently,
\begin{align*}
\log_2
\frac{
|\mathscr F_n^{<n}|
}{
|\mathscr F_n^n|
}
&\leq
-n^2
+
2n\log_2n
+
O(n)\\
&\longrightarrow
-\infty.
\end{align*}
Therefore,
\[
\frac{
|\mathscr F_n^{<n}|
}{
|\mathscr F_n^n|
}
\longrightarrow0.
\]
Finally,
\[
\varphi_n
=
\frac{
|\mathscr F_n^{<n}|
}{
|\mathscr F_n|
}
\leq
\frac{
|\mathscr F_n^{<n}|
}{
|\mathscr F_n^n|
}
\longrightarrow0.
\]
Since $|\mathscr F_n|=|\mathscr F_n^n|+|\mathscr F_n^{<n}|$,
we also obtain
\[
\frac{
|\mathscr F_n^n|
}{
|\mathscr F_n|
}
\longrightarrow1.
\]
\end{proof}


\subsection{Probability of obtaining a singleton core after projection}\label{sec:projection}

The previous subsection shows that the faces of $\cBG(n)$ all of whose
 games have a singleton core form an asymptotically overwhelming proportion
of its facial structure. We now investigate whether this combinatorial
property translates into a probabilistic statement when a game is
projected onto the set of balanced games. As in the closest-game problem
studied in this paper, the value of the grand coalition is preserved.

Take a game $v\in\cG(n)$, put $\alpha=v(N)$, and consider its weighted
projection
\[
v^*=\proj^\gamma_{\cBG_\alpha(n)}(v),
\]
where the strictly positive weights $\gamma_S$ are those of
Problem~\eqref{eq:optimization_problem1}. We are interested in the
probability that $C(v^*)$ is a singleton. The affine correspondence
established in Section~\ref{sec:baco} allows us to work on the single
section $\cBG_0(n)$. Indeed, writing
\[
\bar v=v-\alpha u_{\{n\}},
\]
we have
\begin{equation}\label{eq:centered-projection}
v^*=\proj^\gamma_{\cBG_0(n)}(\bar v)
+\alpha u_{\{n\}},
\qquad
C(v^*)=C\!\left(\proj^\gamma_{\cBG_0(n)}(\bar v)\right)
+\alpha e_n.
\end{equation}
Thus, the singleton-core property is unchanged by this translation.
For equal weights, $\proj^\gamma$ is the ordinary orthogonal projection.

Write $H_n:=\cG_0(n)$ and $K_n:=\cBG_0(n)$. On the vector space $H_n$,
we use the weighted inner product
\[
\langle a,b\rangle_\gamma
:=\sum_{\varnothing\neq S\subsetneq N}\gamma_Sa(S)b(S).
\]
Recall that $\mathscr F_n$ denotes the set of nonempty faces of
$\cBG(n)$, including the improper face. By Section~\ref{sec:baco},
these faces correspond bijectively to those of $K_n$: for every
$\cF\in\mathscr F_n$, let $\cF_0:=\cF\cap H_n$.

We first introduce the regions of $H_n$ whose points project onto the
relative interior of a given face. For $x\in K_n$, its (weighted)
normal cone is defined by
\[
N^\gamma_{K_n}(x)
:=\left\{z\in H_n:
\langle z,y-x\rangle_\gamma\leq0
\quad\text{for every }y\in K_n\right\}.
\]
A standard fact from convex geometry is that the normal cone is constant
on the relative interior of a face; see, for instance,
\cite{schneider2022convexcones}. Hence, for $\cF\in\mathscr F_n$, we may
write $N^\gamma(\cF_0):=N^\gamma_{K_n}(x)$ for any
$x\in\operatorname{relint}(\cF_0)$.

Define the \emph{projection region} associated with $\cF$ by
\[
R_n^\gamma(\cF)
:=\operatorname{relint}(\cF_0)+N^\gamma(\cF_0).
\]
Thus $R_n^\gamma(\cF)$ consists of the games in $H_n$ that can be
written as the sum of a point in the relative interior of $\cF_0$ and
a vector normal to $K_n$ along that face. The standard characterization
of metric projection onto a closed convex set gives
\begin{equation}\label{eq:projection-face}
\proj^\gamma_{K_n}(w)\in\operatorname{relint}(\cF_0)
\quad\Longleftrightarrow\quad
w\in R_n^\gamma(\cF),
\qquad w\in H_n;
\end{equation}
see \cite{bauschke2017convex,schneider2022convexcones}.
Consequently, the projection regions form a disjoint partition of $H_n$:
\begin{equation}\label{eq:projection-partition}
H_n=\bigsqcup_{\cF\in\mathscr F_n}R_n^\gamma(\cF).
\end{equation}
Since $K_n$ is full-dimensional in $H_n$, we also have
$\dim(\cF_0)+\dim N^\gamma(\cF_0)=\dim H_n$. Hence projection
regions associated with faces of every dimension may have positive
probability under an absolutely continuous sampling distribution.
Figure~\ref{fig:mosaico} illustrates this decomposition in dimension two.

Let $\mu_n$ be a probability distribution on $\cG(n)$ and let
$v_n\sim\mu_n$. Define $\bar v_n=v_n-v_n(N)u_{\{n\}}$. For every
$\cF\in\mathscr F_n$, set
\begin{equation}\label{eq:face-probabilities-fixed}
p_n(\cF)
:=\Pr\bigl(\bar v_n\in R_n^\gamma(\cF)\bigr)
=\Pr\!\left(\proj^\gamma_{K_n}(\bar v_n)
\in\operatorname{relint}(\cF_0)\right).
\end{equation}
By Eq.~\eqref{eq:centered-projection}, this is also the probability
that the closest balanced game $v_n^*$ lies in the relative interior of
the corresponding face of $\cBG_{v_n(N)}(n)$. Equation~\eqref{eq:projection-partition}
yields
\begin{equation}\label{eq:sum-face-probabilities}
\sum_{\cF\in\mathscr F_n}p_n(\cF)=1.
\end{equation}
Therefore, the combinatorial proportion of faces with a given property
does not by itself determine the probability that the projection lands
on such a face: this also depends on the distribution of the original games.

To quantify the possible concentration on individual faces, define
\begin{equation}\label{eq:kappa}
\kappa_n
:=|\mathscr F_n|\max_{\cF\in\mathscr F_n}p_n(\cF).
\end{equation}
Since the average face probability equals $1/|\mathscr F_n|$,
Eq.~\eqref{eq:sum-face-probabilities} implies $\kappa_n\geq1$.
Thus $\kappa_n$ measures how much more likely the most probable
projection region is than the average one.

\begin{definition}\label{def:nefpc}
Let $(\mu_n)_n$ be a sequence of probability distributions on $\cG(n)$,
for a fixed choice of strictly positive projection weights at each $n$.
We say that $(\mu_n)_n$ satisfies the \emph{No Excessive Face Probability
Concentration (NEFPC)} condition if
\begin{equation}\label{eq:nefpc}
\log_2\kappa_n=o(n^2),
\qquad\text{or, equivalently,}\qquad
\kappa_n=2^{o(n^2)}.
\end{equation}
\end{definition}
The NEFPC condition does not require the face probabilities to be
asymptotically uniform. It allows the probability of the most likely
projection region to exceed the average by a factor $2^{o(n^2)}$.

The proof of Theorem~\ref{th:faces-singleton-core} gives
\[
\log_2\frac{|\mathscr F_n^{<n}|}{|\mathscr F_n^n|}
\leq -n^2+2n\log_2n+O(n).
\]
Therefore, recalling that
$\varphi_n=|\mathscr F_n^{<n}|/|\mathscr F_n|$, we obtain
\begin{equation}\label{eq:phi-rate}
\varphi_n\leq2^{-n^2+2n\log_2n+O(n)}.
\end{equation}
We can now combine this combinatorial estimate with the concentration
condition.

\begin{theorem}\label{th:projection-singleton}
Let $(\mu_n)_n$ be a sequence of probability distributions on
$\cG(n)$ satisfying the NEFPC condition, and let $v_n\sim\mu_n$.
Define
\[
v_n^*:=\proj^\gamma_{\cBG_{v_n(N)}(n)}(v_n),
\]
where the projection weights are strictly positive. Then
\[
\lim_{n\to\infty}
\Pr\!\left(C(v_n^*)\text{ is a singleton}\right)=1.
\]
\end{theorem}

\begin{proof}
For every realization of $v_n$, Eqs.~\eqref{eq:centered-projection}
and~\eqref{eq:projection-partition} associate the centered projection
$\proj^\gamma_{K_n}(\bar v_n)$ with the relative interior of a unique
face $\cF_0$, where $\cF\in\mathscr F_n$.
If $\cF\in\mathscr F_n^n$, then every game in $\cF_0$ has a singleton
core; Eq.~\eqref{eq:centered-projection} implies that $v_n^*$ does as
well. Consequently,
\begin{align*}
\Pr\!\left(C(v_n^*)\text{ is not a singleton}\right)
&\leq\sum_{\cF\in\mathscr F_n^{<n}}p_n(\cF)\\
&\leq|\mathscr F_n^{<n}|
\max_{\cF\in\mathscr F_n}p_n(\cF)\\
&=\varphi_n\kappa_n.
\end{align*}
By Eqs.~\eqref{eq:phi-rate} and~\eqref{eq:nefpc},
\[
\varphi_n\kappa_n
\leq2^{-n^2+2n\log_2n+O(n)+o(n^2)}
=2^{-n^2+o(n^2)}\longrightarrow0,
\]
since $n\log_2n=o(n^2)$. The result follows.
\end{proof}

\begin{remark}
The translation in Eq.~\eqref{eq:centered-projection} shows that this
result covers both fixed and random values of $v_n(N)$, as well as the
weighted projection employed in CLOBIS. The NEFPC condition, however,
remains an assumption on the sampling distribution and the associated
projection-region probabilities. The simulations of
Section~\ref{sec:simu} illustrate the singleton-core phenomenon but
do not establish NEFPC.
\end{remark}

\section{Concluding remarks}
In this paper we have studied two different problems related to the closest
balanced game to a game with an empty core. First, we have dealt with the
problem of obtaining the closest balanced game, and we have provided a fast
  algorithm to find it, avoiding the exponential complexity of the original
  optimization problem. Second, we have studied whether the closest game has a
  singleton core. Experiments suggest that, as $n$ grows, the probability that
  the closest game has a singleton core approaches one. We have derived a
mathematical proof that the proportion of facets of the polyhedron of balanced
games such that all games in the facet have a singleton core tends to 1 when $n$
grows.  This has been achieved by finding an expression of the asymptotic
  growth of the number of minimal balanced collections, which constitutes by
  itself an important progress in the combinatorial properties of minimal
  balanced collections.
From this result and an asymptotic bound for the proportion of faces
containing non-singleton-core games, we have shown that the probability
that the closest game has a singleton core tends to 1 under the NEFPC
concentration condition.

An important consequence of these achievements and results is that we have
  defined a new solution concept, which may be called the least square core, due
  to its relation with the least core. Our fast algorithm makes this solution
  concept easily computable, up to 20 players. Future research will be focused
  on the properties of the least square core as a solution concept.


\bibliographystyle{plain}

\begin{thebibliography}{}

\bibitem{bauschke2017convex}
H.~H. Bauschke and P.~L. Combettes.
\newblock {\em Convex Analysis and Monotone Operator Theory in Hilbert Spaces}.
\newblock 2nd edition. Springer, Cham, 2017.

\bibitem{schneider2022convexcones}
R.~Schneider.
\newblock {\em Convex Cones: Geometry and Probability}.
\newblock Lecture Notes in Mathematics, volume 2319. Springer, Cham, 2022.



\bibitem{abna23}
T. Abe, S. Nakada.
\newblock Core stability and the {S}hapley value for cooperative games.
\newblock {\em Social Choice and Welfare}, 60:523--543, 2023.


\bibitem{bon63}
O.~Bondareva.
\newblock Some applications of linear programming to the theory of cooperative
  games.
\newblock {\em Problemy Kibernetiki}, 10:119--139, 1963.
\newblock in Russian.

\bibitem{gil53}
D.~Gillies.
\newblock {\em Some theorems on $n$-person games}.
\newblock PhD thesis, Princeton, New Jersey, 1953.

\bibitem{lagrsu23}
D. Laplace Mermoud, M.~Grabisch, P.~Sudhölter.
\newblock Minimal balanced collections and their application to core stability and other topics of game theory.
\newblock {\em Discrete Applied Mathematics}, 341:60--81, 2023.

\bibitem{vnm44}
J. {von Neumann} and O. Morgenstern.
\newblock {\it  Theory of Games and Economic Behavior}.
\newblock Princeton University Press, 1947, 2nd edition.

\bibitem{pel65}
B.~Peleg.
\newblock An inductive method for constructing minimal balanced collections of finite sets.
\newblock {\em Naval Research Logistics Quarterly}, 12, 155--162, 1965.

\bibitem{CloBalGame}
P.~G.~Segador.
\newblock CloBalGame: Implementation of the CLOBIS algorithm in R.
\newblock {\em GitHub repository}, 2025.
\newblock Available at: \url{https://github.com/pgsegador/CloBalGame}.


\bibitem{gagrmi25}
P.~Garcia-Segador, M.~Grabisch, P.~Miranda.
\newblock On the set of balanced games.
\newblock {\em Mathematics of Operations Research}, 50:3, 2047--2072, 2025.

\bibitem{sha67}
L.~S. Shapley.
\newblock On balanced sets and cores.
\newblock {\em Naval Research Logistics Quarterly}, 14:453--460, 1967.

\bibitem{shsh66}
L. S. Shapley,  M. Shubik.
\newblock  Quasi-cores in a monetary economy with nonconvex preferences.
\newblock {\em Econometrica}, 34:805--827, 1966.

\bibitem{shmo50}
J. Sherman, W. J. Morrison.
\newblock Adjustment of an Inverse Matrix Corresponding to a Change in One Element of a Given Matrix.
\newblock {\em Annals of Mathematical Statistic}, 21(1):124-127, 1950.

\bibitem{shubik_mem}
M.~Shubik.
\newblock Game Theory at Princeton, 1949--1955: A Personal Reminiscence.
\newblock {\em History of Political Economy}, 24(Supplement):151--162, 1992.

\bibitem{tik20}
K.~Tikhomirov.
\newblock Singularity of random Bernouilli matrices
\newblock {\em Annals of mathematics}, 191:593--634, 2020.

\bibitem{zhao_core}
J.~Zhao.
\newblock Three little-known and yet still significant contributions of Lloyd Shapley.
\newblock {\em Games and Economic Behavior}, 104:33--37, 2017.


\bibitem{zie95}
G.~Ziegler.
\newblock {\em Lectures on Polytopes}.
\newblock Springer-Verlag, 1995.
\end{thebibliography}

\appendix

\section{Projections on $\cBG_\alpha(3)$}\label{app:n3}
 When $n=3$, the extremal rays are:
\begin{align*}
& w_1=u_{1}-u_{3}, w_2=u_{2}-u_{3}\\
  & r_{12}=-\delta_{12},r_{13}=-\delta_{13},r_{23}=-\delta_{23}\\
  & r_1=\delta_{12}-\delta_3-\delta_{23}\\
  & r_2=\delta_{12}-\delta_3-\delta_{13}\\
  & r_3=-\delta_3.
\end{align*}
The table indicating which extremal rays belong to which facet is identical to
the case of $\cBG(3)$ (see \cite{gagrmi25}). It follows that the facets have exactly the same neighbor
relations as for $\cBG(3)$, hence, the cone $\cC_0^0(n)$ has the same structure of
faces as the cone $\cBG^0(n)$.

Let us detail the  projection on facet $\cB_1=\{1,2,3\}$. Any game $v^*$ in this facet has
the form
\[
v^* = \alpha_{12}r_{12} + \alpha_{13}r_{13} + \alpha_{23}r_{23} + \beta_1w_1 +
\beta_2w_2
\]
with $\alpha_{ij}\geqslant 0$. The coordinates of $v^*$ (given in the order 1,
2, 3, 12, 13, 23, 123) are then
\[
v^* = (\beta_1,\ \beta_2, \ -\beta_1-\beta_2, \ \beta_1+\beta_2-\alpha_{12},
\ -\beta_2-\alpha_{13}, \ -\beta_1-\alpha_{23}, 0).
\]
This game has a point core, given by
$\{(\beta_1,\beta_2,-\beta_1-\beta_2)\}$. Any game $v$ of the form
$v^*+v(N)u_3+ \gamma\lambda^{\cB_1}$ with $\gamma\geqslant 0$ is projected on the facet
$\cB_1$ of $\cBG_{v(N)}(n)$:
\[
v=( \beta_1+\gamma,\ \beta_2+\gamma, \ -\beta_1-\beta_2+\gamma+v(N), \ \beta_1+\beta_2-\alpha_{12},
\ -\beta_2-\alpha_{13}+v(N), \ -\beta_1-\alpha_{23}+v(N),
\ v(N)),
\]
and its projection is $v^*+v(N)u_3$, with point core $\{(\beta_1,\beta_2,-\beta_1-\beta_2+v(N))\}$.

The same procedure can be applied to any facet and any face. Omitting details,
we give hereunder only the result in a table. Each row gives for a face the form
of a game $v$ (value of $v(S)$ for each $S$) whose projection on
$\cBG_{v(N)}(n)$ falls on that face. The expression of its projection is
obtained by letting $\gamma=0$. The last column indicates if the core is reduced to a
singleton (which is simply the coordinates of the projected game for 1,2,3).

\newpage
\rotatebox{90}{\scriptsize
  \begin{tabular}{|c||c|c|c|c|c|c|c||c|}\hline
    face & 1 & 2 & 3 & 12 & 13 & 23 & 123 & point core \\ \hline
    $\cB_1$ & $\beta_1+\gamma$ & $\beta_2+\gamma$ & $-\beta_1-\beta_2+v(N)+\gamma$ & $\beta_1+\beta_2-\alpha_{12}$ &
    $-\beta_2+v(N)-\alpha_{13}$ & $-\beta_1+v(N)-\alpha_{23}$ & $v(N)$ & yes\\
    $\cB_2$ & $\beta_1+\gamma$ & $\beta_2$ & $-\beta_1-\beta_2+v(N)-\alpha_2-\alpha_3$ & $\beta_1+\beta_2-\alpha_{12}+\alpha_2$ &
    $-\beta_2+v(N)-\alpha_{13}-\alpha_2$ & $-\beta_1+v(N)+\gamma$ & $v(N)$ & no \\
    $\cB_3$ & $\beta_1$ & $\beta_2+\gamma$ &
    $-\beta_1-\beta_2+v(N)-\alpha_1-\alpha_3$ &
    $\beta_1+\beta_2-\alpha_{12}+\alpha_1$ &    $-\beta_2+v(N)+\gamma$ &
    $-\beta_1+v(N)-\alpha_{23}-\alpha_1$ & $v(N)$ & no \\
    $\cB_4$ & $\beta_1$ & $\beta_2$ & $-\beta_1-\beta_2+v(N)+\gamma-\alpha_1-\alpha_2$ & $\beta_1+\beta_2+\gamma+\alpha_1+\alpha_2$ &
    $-\beta_2+v(N)-\alpha_{13}-\alpha_2$ & $-\beta_1+v(N)-\alpha_{23}-\alpha_1$
    & $v(N)$ & no \\
    $\cB_5$ & $\beta_1$ & $\beta_2$ & $-\beta_1-\beta_2+v(N)-\alpha_1-\alpha_2-\alpha_3$ & $\beta_1+\beta_2+\gamma+\alpha_1+\alpha_2$ &
    $-\beta_2+v(N)+\gamma-\alpha_2$ & $-\beta_1+v(N)+\gamma-\alpha_1$ & $v(N)$ &
    yes \\
    $\cB_1\cap\cB_2$ & $\beta_1+\gamma_1+\gamma_2$ & $\beta_2+\gamma_1$ &
    $-\beta_1-\beta_2+v(N)+\gamma_1$ & $\beta_1+\beta_2-\alpha_{12}$ &
    $-\beta_2+v(N)-\alpha_{13}$ & $-\beta_1+v(N)+\gamma_2$ & $v(N)$ & yes \\
    $\cB_1\cap\cB_3$ & $\beta_1+\gamma_1$ & $\beta_2+\gamma_1+\gamma_3$ & $-\beta_1-\beta_2+v(N)+\gamma_1$ &   $\beta_1+\beta_2-\alpha_{12}$ &
    $-\beta_2+v(N)+\gamma_3$ & $-\beta_1+v(N)-\alpha_{23}$ & $v(N)$ & yes\\
    $\cB_1\cap\cB_4$ & $\beta_1+\gamma_1$ & $\beta_2+\gamma_1$ & $-\beta_1-\beta_2+v(N)+\gamma_1+\gamma_4$ &
    $\beta_1+\beta_2+\gamma_4$ &
    $-\beta_2+v(N)-\alpha_{13}$ & $-\beta_1+v(N)-\alpha_{23}$ & $v(N)$ & yes\\
    $\cB_2\cap\cB_5$ & $\beta_1+\gamma_2$ & $\beta_2$ & $-\beta_1-\beta_2+v(N)-\alpha_2-\alpha_3$ &
    $\beta_1+\beta_2+\gamma_5+\alpha_2$ &
    $-\beta_2+v(N)+\gamma_5-\alpha_2$ & $-\beta_1+v(N)+\gamma_2+\gamma_5$ &
    $v(N)$ & yes\\
    $\cB_3\cap\cB_5$ & $\beta_1$ & $\beta_2+\gamma_3$ & $-\beta_1-\beta_2+v(N)-\alpha_1-\alpha_3$ &
    $\beta_1+\beta_2+\gamma_5+\alpha_1$ &
    $-\beta_2+v(N)+\gamma_3+\gamma_5$ & $-\beta_1+v(N)+\gamma_5-\alpha_1$ &
    $v(N)$ & yes \\
    $\cB_4\cap\cB_5$ & $\beta_1$ & $\beta_2$ & $-\beta_1-\beta_2+v(N)+\gamma_4-\alpha_1-\alpha_2$ &
    $\beta_1+\beta_2+\gamma_4+\gamma_5+\alpha_1+\alpha_2$ &
    $-\beta_2+v(N)+\gamma_5-\alpha_2$ & $-\beta_1+v(N)+\gamma_5-\alpha_1$ &
    $v(N)$ & yes \\
    $\cB_2\cap\cB_3$ & $\beta_1+\gamma_2$ & $\beta_2+\gamma_3$ & $-\beta_1-\beta_2+v(N)-\alpha_3$ &
    $\beta_1+\beta_2-\alpha_{12}$ &
    $-\beta_2+v(N)+\gamma_3$ & $-\beta_1+v(N)+\gamma_2$ & $v(N)$ & yes\\
    $\cB_2\cap\cB_4$ & $\beta_1+\gamma_2$ & $\beta_2$ &
    $-\beta_1-\beta_2+v(N)+\gamma_4-\alpha_2$ &
    $\beta_1+\beta_2+\gamma_4+\alpha_2$ &
    $-\beta_2+v(N)-\alpha_{13}-\alpha_2$ & $-\beta_1+v(N)+\gamma_2$ & $v(N)$ & yes\\
    $\cB_3\cap\cB_4$ & $\beta_1$ & $\beta_2+\gamma_3$ & $-\beta_1-\beta_2+v(N)+\gamma_4-\alpha_1$ &
    $\beta_1+\beta_2+\gamma_4+\alpha_1$ &
    $-\beta_2+v(N)+\gamma_3$ & $-\beta_1+v(N)-\alpha_{23}-\alpha_1$ & $v(N)$ & yes\\
    $\cB_1\cap\cB_2\cap\cB_3$ & $\beta_1+\gamma_1+\gamma_2$ & $\beta_2+\gamma_1+\gamma_3$ & $-\beta_1-\beta_2+v(N)+\gamma_1$ &
    $\beta_1+\beta_2-\alpha_{12}$ &
    $-\beta_2+v(N)+\gamma_3$ & $-\beta_1+v(N)+\gamma_2$ & $v(N)$ & yes\\
    $\cB_1\cap\cB_2\cap\cB_4$ & $\beta_1+\gamma_1+\gamma_2$ & $\beta_2+\gamma_1$ & $-\beta_1-\beta_2+v(N)+\gamma_1+\gamma_4$
    & $\beta_1+\beta_2+\gamma_4$ &
    $-\beta_2+v(N)-\alpha_{13}$ & $-\beta_1+v(N)+\gamma_2$ & $v(N)$ & yes\\
    $\cB_1\cap\cB_3\cap\cB_4$ & $\beta_1+\gamma_1$ & $\beta_2+\gamma_1+\gamma_3$ & $-\beta_1-\beta_2+v(N)+\gamma_1+\gamma_4$
    & $\beta_1+\beta_2+\gamma_4$ &
    $-\beta_2+v(N)+\gamma_3$ & $-\beta_1+v(N)-\alpha_{23}$ & $v(N)$ & yes\\
    $\cB_2\cap\cB_3\cap\cB_5$ & $\beta_1+\gamma_2$ & $\beta_2+\gamma_3$ & $-\beta_1-\beta_2+v(N)-\alpha_3$
    & $\beta_1+\beta_2+\gamma_5$ &
    $-\beta_2+v(N)+\gamma_3+\gamma_5$ & $-\beta_1+v(N)+\gamma_2+\gamma_5$ & $v(N)$ & yes\\
    $\cB_2\cap\cB_4\cap\cB_5$ & $\beta_1+\gamma_2$ & $\beta_2$ & $-\beta_1-\beta_2+v(N)+\gamma_4-\alpha_2$
    & $\beta_1+\beta_2+\gamma_4+\gamma_5+\alpha_2$ &
    $-\beta_2+v(N)+\gamma_5-\alpha_2$ & $-\beta_1+v(N)+\gamma_2+\gamma_5$ &
    $v(N)$ & yes \\
    $\cB_3\cap\cB_4\cap\cB_5$ & $\beta_1$ & $\beta_2+\gamma_3$ & $-\beta_1-\beta_2+v(N)+\gamma_4-\alpha_1$
    & $\beta_1+\beta_2+\gamma_4+\gamma_5+\alpha_1$ &
    $-\beta_2+v(N)+\gamma_3+\gamma_5$ & $-\beta_1+v(N)+\gamma_5-\alpha_1$ &
    $v(N)$ & yes \\
    \hline
  \end{tabular}
  }

\end{document}